\documentclass[11pt]{amsart}
\usepackage[
  left=2.5cm,
  right=2.5cm,
  top=3cm,
  bottom=3cm
]{geometry}

\usepackage{amsmath,amssymb,amsthm,mathrsfs,url,bbm}
\usepackage[usenames,dvipsnames]{color}
\usepackage{microtype,needspace}
\usepackage[colorlinks=true,linkcolor=MidnightBlue,citecolor=MidnightBlue,urlcolor=MidnightBlue]{hyperref}
\usepackage{braket}
\hypersetup{
  pdftitle={Computing the free energy of quantum Coulomb gases via quantum Gibbs sampling},
  pdfauthor={Simon Becker, Cambyse Rouze, Robert Salzmann},
  pdfsubject={Coulomb approximation and spectral gaps for quantum Gibbs samplers},
  pdfkeywords={Coulomb interaction, Gibbs sampling, spectral gap, free energy}
}

\def\smallsection#1{\smallskip\noindent\textbf{#1}.}

\newcommand{\kb}[1]{|#1\rangle\!\langle#1|}

\newcommand{\eps}{\varepsilon}

\renewcommand{\phi}{\varphi}

\theoremstyle{plain}
\newtheorem{theorem}{Theorem}[section]
\newtheorem{lemma}[theorem]{Lemma}
\newtheorem{proposition}[theorem]{Proposition}
\newtheorem{corollary}[theorem]{Corollary}

\numberwithin{equation}{section}

\begin{document}
\title[Gibbs sampling for Coulomb quantum gases]{Computing the free energy of quantum Coulomb gases via quantum Gibbs sampling}

\author{Simon Becker}
\email{simon.becker@unibocconi.it}
\address{Uni Bocconi, 20136, Milan, Italy}

\author{Cambyse Rouz\'e}
\email{cambyse.rouze@inria.fr}
\address{Inria, T\'el\'ecom Paris -- LTCI, Institut Polytechnique de Paris, 91120 Palaiseau, France}

\author{Robert Salzmann}
\email{robert.salzmann@rwth-aachen.de}
\address{Institute for Quantum Information, RWTH Aachen University, Germany}

\begin{abstract}
We develop a quantum Gibbs-sampling approach to free-energy estimation for
trapped systems of distinguishable particles with Coulomb interactions in dimensions $d\in\{2,3\}$.  The singular
infinite-dimensional problem is reduced to a finite-rank low-energy
truncation of the interaction, the truncated Gibbs state is prepared by a
quantum Markov semigroup, and thermodynamic integration gives the free
energy.  We obtain explicit truncation rates, quantitative spectral-gap
estimates, and a finite-dimensional circuit implementation.  In
particular, every fixed interaction truncation has a positive gap, while
weak truncated interactions admit a particle-number-independent lower
bound.  The circuit complexity is conditional on a certified gap lower
bound at the accuracy-dependent interaction cutoff.
\end{abstract}

\maketitle
\tableofcontents

\section{Introduction}

Free energy is a basic quantity in equilibrium statistical mechanics.
Free-energy differences determine equilibrium weights and compare
competing states of interacting many-body systems.  For quantum particles,
direct computation is
difficult when the Hilbert space is infinite dimensional and the
interaction is singular.

Our strategy has four steps.  We first replace the Coulomb interaction by
a controlled low-energy truncation.  We then sample from the Gibbs state
of the truncated Hamiltonian using a quantum Markov semigroup, prove a
positive spectral gap, and finally combine Gibbs-state preparation with
thermodynamic integration to recover the free energy.

\subsection{General setting}

We write $\mathbb R$ and $\mathbb C$ for the real and complex numbers.
Fix integers $n\ge1$ and $d\in\{2,3\}$.  For particle $i$, $1\le i\le n$,
write $x_i=(x_{i1},\ldots,x_{id})\in\mathbb R^d$ and
$x=(x_1,\ldots,x_n)$.  The Hilbert space is
\[
 \mathcal H_n=L^2((\mathbb R^d)^n).
\]
Particles are labelled and distinguishable: no bosonic or fermionic
symmetry sector is imposed, and the individual ladder jumps below do not
preserve such sectors.
We write $|\cdot|$ for the Euclidean norm and
$\Delta=\sum_{k=1}^d\partial_{x_k}^2$ for the one-particle Laplacian.
The one-particle harmonic oscillator is
$h=-\Delta+|x|^2$ on $L^2(\mathbb R^d)$.  If $h_i$ acts on particle
$i$, we set
\begin{equation}
\label{eq:model-hamiltonian}
 H_{0,n}=\sum_{i=1}^n h_i,\qquad
 W_n=\sum_{1\le i<j\le n}\alpha_{n,i,j}W_{ij},\qquad H_n=H_{0,n}+W_n.
\end{equation}
Here the real numbers $\alpha_{n,i,j}$ are the pair couplings and
$W_{ij}$ is multiplication, as a quadratic form, by
\[
 w_d(x_i-x_j),\qquad
 w_2(y)=-\log|y|,\qquad
 w_3(y)=|y|^{-1},\qquad y\in\mathbb R^d\setminus\{0\}.
\]
Throughout, $\log$ is the natural logarithm.  The quadratic trap is kept
throughout the paper.  It gives confinement at infinity, while the
Coulomb singularity is treated by form estimates.

For inverse temperature $\beta>0$ and a Hamiltonian $H$ with finite heat
trace, $\operatorname{Tr}$ denotes the operator trace and we define
\[
 \mathcal Z_\beta(H)=\operatorname{Tr}(e^{-\beta H}),\qquad
 \sigma_\beta(H)=\frac{e^{-\beta H}}{\mathcal Z_\beta(H)},\qquad
 F_\beta(H)=-\beta^{-1}\log\mathcal Z_\beta(H).
\]
The form estimates in Section~\ref{sec:partition-functions} show that the
heat trace is finite for all Hamiltonians used below.  A density operator
is a positive trace-class operator with trace one.

For a vector $u$, $\|u\|$ denotes its Hilbert-space norm; for an operator
$T$, $\|T\|$ is the operator norm and $T^\dagger=T^*$ is its adjoint.
We write $|T|=(T^*T)^{1/2}$.  On a Hilbert space $\mathcal H$, the
Schatten classes $\mathscr T_p(\mathcal H)$, $p=1,2$, have norm
$\|T\|_p=(\operatorname{Tr}|T|^p)^{1/p}$; these are the trace-class and
Hilbert--Schmidt operators.  We omit $\mathcal H$ when it is clear and use
\[
 \langle X,Y\rangle_2=\operatorname{Tr}(X^*Y)
\]
for the Hilbert--Schmidt inner product, linear in $Y$.  The symbol
$\otimes$ denotes the Hilbert-space tensor product.

\subsection{Hamiltonian truncation}

Order the one-particle Hermite functions by nondecreasing energy,
lexicographically at ties, and write $\lambda_m$ for the energy of the
$m$th function in this ordering.  For an integer $M\ge1$, let $P_M$ be
the orthogonal projection onto the first $M$ such functions, and let
\begin{equation}
\label{eq:global-interaction-definition}
 \Pi_M=P_M^{\otimes n},\qquad
 W_{n,M}=\Pi_MW_n\Pi_M,\qquad
 H_{n,M}=H_{0,n}+W_{n,M}.
\end{equation}
Only the interaction is truncated; the oscillator Hamiltonian keeps its
full state space.  We set
\[
 \alpha_n^{\max}=\max_{i<j}|\alpha_{n,i,j}|,
\]
with \(\alpha_n^{\max}=0\) for \(n=1\).

\begin{theorem}\label{thm:coulomb-approximation}
Fix $d\in\{2,3\}$ and $\beta>0$, and define
$Q_n(a)=n^{5/2}a+n^{9/2}a^3$ for $a\ge0$.
There are computable constants $C_F,C_\rho$, independent of $n,M$ and
the couplings, such that
\[
 \begin{aligned}
 |F_\beta(H_n)-F_\beta(H_{n,M})|
 &\le C_F Q_n(\alpha_n^{\max})M^{-1/(2d)},\\
 \|\sigma_\beta(H_n)-\sigma_\beta(H_{n,M})\|_1
 &\le C_\rho Q_n(\alpha_n^{\max})^{1/2}M^{-1/(4d)}.
 \end{aligned}
\]
\end{theorem}

For the sharper coupling dependence, we extend the coefficients
symmetrically, set $\alpha_{n,i,i}=0$, and define
\begin{equation}
\label{eq:coupling-sums}
 A_n=\max_i\sum_{j\ne i}|\alpha_{n,i,j}|,\qquad
 B_n=\sum_{i<j}|\alpha_{n,i,j}|.
\end{equation}
With
\begin{equation}
\label{eq:coupling-envelope}
 \Xi_n=\sqrt{A_nB_n}\bigl[n+(A_n+1)B_n\bigr]
 \le Q_n(\alpha_n^{\max}),
\end{equation}
the same constants satisfy
\begin{align}
 |F_\beta(H_n)-F_\beta(H_{n,M})|
 &\le C_F\Xi_n M^{-1/(2d)},\label{eq:main-free-energy-error}\\
 \|\sigma_\beta(H_n)-\sigma_\beta(H_{n,M})\|_1
 &\le C_\rho\Xi_n^{1/2}M^{-1/(4d)},\label{eq:main-gibbs-error}\\
 D(\sigma_\beta(H_n)\|\sigma_\beta(H_{n,M}))
 +D(\sigma_\beta(H_{n,M})\|\sigma_\beta(H_n))
 &\le C_\rho^2\Xi_n M^{-1/(2d)}.\label{eq:main-entropy-error}
\end{align}
Here $D$ is relative entropy, defined in Section~\ref{sec:partition-functions}.
The last estimate compares the two faithful infinite-dimensional Gibbs
states; it does not assert symmetric relative-entropy accuracy for a
finite-register output.
The explicit constants are given in
\eqref{eq:explicit-approximation-constants}.
For $d=3$ and $\alpha_{n,i,j}\ge0$, all three estimates also hold with
\begin{equation}\label{eq:repulsive-envelope}
 \Xi_n^+=\sqrt{A_nB_n}(n+B_n)
\end{equation}
in place of $\Xi_n$, by Proposition~\ref{prop:thermal-energy} below.
Write $\widehat\Xi_n=\Xi_n^+$ in this repulsive case and
$\widehat\Xi_n=\Xi_n$ otherwise.
If each particle has at most $\Delta$ nonzero couplings of size at most
$a$, then $A_n\le\Delta a$, $B_n\le n\Delta a/2$, and
\[
 \Xi_n\le\frac{n^{3/2}\Delta a}{\sqrt2}
       \left[1+\frac{\Delta a(1+\Delta a)}2\right].
\]
Thus bounded-degree interactions with bounded couplings have
$\Xi_n=\mathcal O(n^{3/2})$.
Thermodynamic integration then expresses the truncated free-energy
difference through Gibbs expectations along
$H(s)=H_{0,n}+sW_{n,M}$; the precise identity is stated later in
\eqref{eq:thermodynamic-integration}.

\subsection{Quantum Gibbs sampling}

To prepare $\sigma_\beta(H_{n,M})$, we use the quantum Markov-semigroup
construction of \cite{BeckerRouzeSalzmannToAppearcmp}.  For
$1\le i\le n$ and $1\le k\le d$, let $\partial_{x_{ik}}$ denote
differentiation in coordinate $x_{ik}$ and define the oscillator ladder
operators
\[
 a_{ik}=\frac{x_{ik}+\partial_{x_{ik}}}{\sqrt2},\qquad
 a_{ik}^\dagger=\frac{x_{ik}-\partial_{x_{ik}}}{\sqrt2},
 \qquad
 \mathcal A=\{a_{ik},a_{ik}^\dagger:1\le i\le n,\ 1\le k\le d\}.
\]
We choose a filter $f$ in the Schwartz space
$\mathcal S(\mathbb R;\mathbb C)$ of smooth complex-valued functions whose
derivatives decay faster than every inverse power.  For a real frequency
$\nu$, define
$\widehat f(\nu)=\int_{\mathbb R}f(t)e^{it\nu}\,dt$.  In exponentials,
$i=\sqrt{-1}$ is the imaginary unit; an overline denotes complex
conjugation, and $\|\cdot\|_\infty$ denotes the essential supremum norm of
a scalar function.  We assume
\begin{equation}
\label{eq:KMSFilterFunction}
 \overline{\widehat f(\nu)}=e^{-\beta\nu/2}\widehat f(-\nu),\qquad
 \widehat f(\nu)\ne0,\qquad
 \|\widehat f\|_\infty+\|e^{\beta(\cdot)/2}\widehat f\|_\infty<\infty
 \quad(\nu\in\mathbb R).
\end{equation}
For a Hamiltonian \(H\), \(A\in\mathcal A\), and a test vector \(\psi\) in the natural domain, the filtered jump is
\[
 L_H^A\psi=\int_{\mathbb R}f(t)e^{itH}Ae^{-itH}\psi\,dt.
\]
The full trace-class generator is denoted by
\(\mathcal L_{\sigma_E,H}\), where \(0<\sigma_E<\infty\) is the Gaussian
energy-resolution parameter; \(\sigma_E=\infty\) denotes the coherent
version.  Its self-adjoint Hilbert--Schmidt realization is
\(L_{\sigma_E,H}\).

The precise domains, integral formulas, and invariant-state statement are
given in Section~\ref{sec:sampler-definition}; the infinite-dimensional
generation theory is taken from \cite{BeckerRouzeSalzmannToAppearcmp}.

\subsection{Mixing time via spectral gap analysis}

Let $g_H=\sigma_\beta(H)^{1/2}$ and define its Hilbert--Schmidt
orthogonal complement by
\[
 g_H^\perp=\{X\in\mathscr T_2:\langle g_H,X\rangle_2=0\}.
\]
The gap $\operatorname{gap}(L_{\sigma_E,H})$ is the Poincar\'e constant
of $-L_{\sigma_E,H}$ on $g_H^\perp$.  For the free oscillator,
\begin{equation}
\label{eq:oscillator-gap-definition}
 \lambda_{\rm OU}=
 \frac{|\widehat f(-2)|^2-|\widehat f(2)|^2}{2}
 =\frac{1-e^{-2\beta}}2|\widehat f(-2)|^2>0.
\end{equation}
For a target trace-norm error $0<\varepsilon<1$, a positive lower bound
$\gamma$ on the gap gives exponential convergence of warm starts and a
mixing time of order $\gamma^{-1}\log(1/\varepsilon)$.

\begin{theorem}\label{thm:cutoff-gaps}
Fix $n,M$ and a bounded self-adjoint operator $W=\Pi_MW\Pi_M$.
The coherent generator $L_{\infty,H_{0,n}+W}$ has compact resolvent and
kernel $\mathbb Cg_{H_{0,n}+W}$, the span of its Gibbs vector.  Along $H(s)=H_{0,n}+sW$, its gap is
continuous and has a strictly positive minimum on $[0,1]$.
For every $s\in[0,1]$ and $0<\sigma_E<\infty$, the averaged sampler has
at least this minimum gap.

For fixed $d,\beta,M,f$, there is an explicit $b_M^*>0$, independent of
$n$, such that
\[
 \|W\|\le b_M^*
 \quad\Longrightarrow\quad
 \operatorname{gap}(L_{\sigma_E,H_{0,n}+sW})\ge\lambda_{\rm OU}/4
 \quad(0\le s\le1,\ 0<\sigma_E\le\infty).
\]
In particular, an explicit $\varepsilon_M>0$, independent of $n$, makes
$|\alpha_{n,i,j}|\le\varepsilon_M/n^2$ sufficient for this bound with
$W=W_{n,M}$.
\end{theorem}

Section~\ref{section:Pert_Theory} also proves a cutoff-free estimate for
bounded perturbations with small ladder commutators, yielding a uniform
gap for bounded Lipschitz pair potentials under a sufficiently small
\(n^{-3/2}\) pair-coupling scale.

\subsection{Circuit implementation}

For the circuit bounds we use the Gaussian filter
\[
 \widehat f_\ell(\nu)
 =\exp\!\left(-\frac{\nu^2}{4\ell^2}-\frac{\beta\nu}{4}\right),
 \qquad \ell>0,
\]
and the matrix and scalar-filter access model stated at the beginning of
Section~\ref{sec:algorithm-input}.  A certified positive lower bound on
the relevant relaxation gap is part of the input.  For nonnegative
quantities $X,Y$, we write $X\lesssim Y$ or $X=\mathcal O(Y)$ when
$X\le CY$ for a constant $C$ depending only on fixed parameters, and
$X\asymp Y$ when both $X\lesssim Y$ and $Y\lesssim X$.  The notation
$\operatorname{poly}$ denotes an unspecified polynomial in its displayed
arguments, while $\widetilde{\mathcal O}$ additionally suppresses
polylogarithmic factors.  Subscripts on $\lesssim$ or $\mathcal O$ list
the parameters on which the implicit constants may depend.

\begin{theorem}\label{thm:algorithms-intro}
Assume the stated input model and $\alpha_n^{\max}\le\operatorname{poly}(n)$.
Fix an error tolerance $0<\varepsilon<1$ and a failure probability
$0<\delta\le1/2$.  With sufficiently large constants depending on
$d,\beta$, choose the state and free-energy interaction ranks
\[
 M_\rho\asymp1+\left(\frac{\widehat\Xi_n}{\varepsilon^2}\right)^{2d},
 \qquad
 M_F\asymp1+\left(\frac{\widehat\Xi_n}{\varepsilon}\right)^{2d}.
\]
Given $0<\gamma\le\min\{1,\operatorname{gap}
(L_{\sigma_E,H_{n,M_\rho}})\}$, the state $\sigma_\beta(H_n)$ can be
prepared to trace-norm error $\varepsilon$, using
\[
 \mathcal O\left(n\left[\log(2+n/\varepsilon)
                  +\log\log(e^2+1/\gamma)\right]\right)
\]
system qubits and circuit depth
$\widetilde{\mathcal O}(\gamma^{-1}\operatorname{poly}(n,1/\varepsilon))$.

Given
\[
 0<\gamma_*\le\min\left\{1,\min_{0\le s\le1}
 \operatorname{gap}(L_{\infty,H_{0,n}+sW_{n,M_F}})\right\},
\]
the free energy $F_\beta(H_n)$ can be estimated to error $\varepsilon$
with failure probability at most $\delta$, using the same system-qubit
bound with $\gamma_*$ in place of $\gamma$ and runtime
\[
 \widetilde{\mathcal O}\bigl(\gamma_*^{-1}\log(2/\delta)
                         \operatorname{poly}(n,1/\varepsilon)\bigr).
\]
Writing $v_{M_F}=\|(P_{M_F}\otimes P_{M_F})W_{12}
(P_{M_F}\otimes P_{M_F})\|$, the total number of independent prepared
states, each followed by one selected pair measurement, is
\[
 \mathcal O\left(\left[1+\frac{B_n^2v_{M_F}^2}{\varepsilon^2}\right]
                           \log\frac2\delta\right).
\]
\end{theorem}

The randomized thermodynamic-integration estimator in the proof samples an
integration point and an interaction pair jointly, so the bound in
Theorem~\ref{thm:algorithms-intro} has no additional multiplicative factor
for all pairs or all grid points.

\subsection{Related work}

\smallsection{Classical systems}
Classical free-energy methods combine equilibrium sampling and
thermodynamic integration
\cite{rousset2010free,chafai2019simulating,chafai2021coulomb}.
Coulomb singularities create additional difficulties for equilibration;
see \cite{Bolley2018,chafai2021aspects,lu2020geometric,Duong2024}.
For weakly coupled repulsive mean-field Coulomb gases in dimensions
\(d\ge3\), Becker and Menegaki prove a particle-number-uniform
Poincar\'e inequality, and hence a uniform spectral gap for the
corresponding overdamped Langevin dynamics
\cite{BeckerMenegaki2026}.  That result concerns a classical diffusion;
the dynamics studied here is a quantum Gibbs sampler.

\smallsection{Quantum Gibbs sampling}
Early quantum algorithms for thermal states include quantum simulated
annealing, partition-function estimation, and the quantum Metropolis
algorithm
\cite{SommaEtAl2008,PoulinWocjan2009,TemmeEtAl2011,ChowdhurySomma2017}.
Quantum belief propagation and correlation-based preparation provide a
different approach for local systems
\cite{Hastings2007,BrandaoKastoryano2019}.
The mathematical theory of quantum detailed balance and mixing includes
\cite{FagnolaUmanita2010,TemmeEtAl2010,KastoryanoTemme2013,TemmeDavies2013},
while commuting Gibbs samplers were studied in
\cite{KastoryanoBrandao2016}.

Recent implementable samplers based on exact or approximate detailed
balance include
\cite{RallWangWocjan2023,chen2023efficient,ding2025efficient,gilyen2024quantum}.
Rapid-mixing results for high-temperature local systems, random sparse
Hamiltonians, and weakly interacting fermions appear in
\cite{RouzeFrancaAlhambra2026,RamkumarSoleimanifar2025,TongZhan2025,SmidEtAl2025}.
Related system--bath and Bose--Hubbard results are
\cite{slezak2026polynomial,BeckerRouzeSalzmannBose}.
The infinite-dimensional construction and circuit implementation used
here come from \cite{BeckerRouzeSalzmannToAppearcmp}; our reference
quadratic sampler is related to the quantum Ornstein--Uhlenbeck dynamics
of \cite{OU}.

For each fixed particle number and finite-rank interaction cutoff, hypercontractivity and exponential relative-entropy decay for Gaussian-filter samplers are proved in~\cite{BeckerRouzeSalzmannHypercontractivity2026}.

The rest of the paper follows this four-part division.
Section~\ref{sec:fin_rank} proves the finite-rank
approximation, Section~\ref{section:Pert_Theory} proves the gap estimates,
and Section~\ref{sec:algorithm-input} gives the circuit implementation
and free-energy algorithm.

\section{Finite-rank perturbation analysis}
\label{sec:fin_rank}
\subsection{Partition functions}
\label{sec:partition-functions}
We first relate errors in the interaction to errors in the free energy
and then estimate the part of the Coulomb interaction removed by the
low-energy projection.

\smallsection{Gibbs stability}
We use $\mathbb N=\{1,2,\ldots\}$, $\mathbb N_0=\{0,1,\ldots\}$,
and $\mathbb Z$ for the integers.
Inner products are linear in their second argument.  The identity
operator is $I$, and scalar shifts of operators mean multiples of $I$.
For an orthogonal projection $P$, we set $P^\perp=I-P$.
For an operator $T$, $D(T)$ is its domain, $\operatorname{Ran}T$ its
range, and $\ker T$ its kernel.  For a lower-bounded self-adjoint $H_0$,
its form domain is
$Q(H_0)=D((H_0+cI)^{1/2})$, where $c$ is any real number such that
$H_0+cI\ge I$.  We write $B[\psi]=B(\psi,\psi)$ for a quadratic form
$B$ and use $\langle\psi,H_0\psi\rangle$ for the form value on $Q(H_0)$.
Constants $C,c$ are positive and may change between estimates; subscripts
list their allowed parameter dependence.
For density operators $\rho,\sigma$, we use the extended relative entropy
\(D(\rho\|\sigma)=\operatorname{Tr}\bigl(\rho(\log\rho-\log\sigma)\bigr),\)
with value $+\infty$ when the expression is not finite, in particular
when $\ker\sigma\not\subseteq\ker\rho$.
Under the assumptions of the next lemma, all the required energy and
entropy terms are finite.

\begin{lemma}\label{lem:gibbs-stability}
Let $H_0\ge0$ be self-adjoint on a Hilbert space $\mathcal H$, with
$\operatorname{Tr}e^{-tH_0}<\infty$ for every $t>0$.  Let $W,W'$ be
symmetric forms on $Q(H_0)$ with relative form bounds below one, and set
\(H=H_0+W,\qquad H'=H_0+W',\qquad \Delta=H-H'.\)
Suppose, for some $a,c\ge0$, that
\[
 |\Delta[\psi]|\le a\langle\psi,H_0\psi\rangle+c\|\psi\|^2
 \qquad(\psi\in Q(H_0)).
\]
For $\rho=\sigma_\beta(H)$, $\rho'=\sigma_\beta(H')$, and
$U_0(\tau)=\operatorname{Tr}(\tau H_0)$ for a density operator $\tau$
of finite $H_0$-energy, one has
\begin{align}
 |F_\beta(H)-F_\beta(H')|
 &\le a\max\{U_0(\rho),U_0(\rho')\}+c,
       \label{eq:gibbs-free-stability}\\
 D(\rho\|\rho')+D(\rho'\|\rho)
 &=\beta\operatorname{Tr}((\rho'-\rho)\Delta)
 \le\beta\bigl[a(U_0(\rho)+U_0(\rho'))+2c\bigr],
       \label{eq:gibbs-entropy-stability}\\
 \|\rho-\rho'\|_1^2
 &\le\beta\bigl[a(U_0(\rho)+U_0(\rho'))+2c\bigr].
       \label{eq:gibbs-trace-stability}
\end{align}
\end{lemma}

\begin{proof}
The form bounds give self-adjoint operators with common form domain
$Q(H_0)$ and constants $a_0>0$ and $C>0$ such that
\[
 H\ge a_0H_0-C,
 \qquad
 H'\ge a_0H_0-C.
\]
Since $\operatorname{Tr}e^{-tH_0}<\infty$, the resolvent of $H_0$ is
compact.  The same is therefore true for $H$ and $H'$.  If
$\lambda_k(T)$ denotes the $k$th eigenvalue of a lower-bounded operator
$T$, counted with multiplicity, the min--max principle gives
\[
 \lambda_k(H),\ \lambda_k(H')
 \ge a_0\lambda_k(H_0)-C.
\]
Hence, for every $t>0$,
\[
 \begin{aligned}
 \operatorname{Tr}e^{-tH}
 &=\sum_k e^{-t\lambda_k(H)}
 \le e^{tC}\sum_k e^{-ta_0\lambda_k(H_0)}
 =e^{tC}\operatorname{Tr}e^{-ta_0H_0}<\infty,\\
 \operatorname{Tr}e^{-tH'}
 &\le e^{tC}\operatorname{Tr}e^{-ta_0H_0}<\infty.
 \end{aligned}
\]
The same comparison, applied with a slightly smaller positive time,
shows that the Gibbs states have finite $H_0$-energy.  For example,
for $0<s<\beta a_0$,
\[
 \operatorname{Tr}(H_0e^{-\beta H})
 \le C_{s,a_0,C}\operatorname{Tr}e^{-sH_0}<\infty,
\]
and similarly for $H'.$
The Gibbs variational principle, equivalently the Bogoliubov inequality
\cite{Dereziski2003}, gives
\begin{equation}
\label{eq:gibbs-variational-bracket}
 \operatorname{Tr}(\rho\Delta)
 \le F_\beta(H)-F_\beta(H')
 \le\operatorname{Tr}(\rho'\Delta).
\end{equation}
Write
\[
 \rho=\sum_k p_k|u_k\rangle\langle u_k|,
 \qquad
 \rho'=\sum_k p_k'|u_k'\rangle\langle u_k'|.
\]
The form bound on $\Delta$ gives
\[
 \begin{aligned}
 |\operatorname{Tr}(\rho\Delta)|
 &\le \sum_k p_k|\Delta[u_k]|
 \le a\sum_kp_k\langle u_k,H_0u_k\rangle+c\sum_kp_k\\
 &=aU_0(\rho)+c,
 \end{aligned}
\]
and, in the same way,
\[
 |\operatorname{Tr}(\rho'\Delta)|\le aU_0(\rho')+c.
\]
Combining these two estimates with
\eqref{eq:gibbs-variational-bracket} yields
\[
 |F_\beta(H)-F_\beta(H')|
 \le a\max\{U_0(\rho),U_0(\rho')\}+c,
\]
which is \eqref{eq:gibbs-free-stability}.  Since
$\log\rho=-\beta H-\log\mathcal Z_\beta(H)$,
\[
 \begin{aligned}
 D(\rho\|\rho')&=\beta\bigl[F_\beta(H)-F_\beta(H')
                                     -\operatorname{Tr}(\rho\Delta)\bigr],\\
 D(\rho'\|\rho)&=\beta\bigl[\operatorname{Tr}(\rho'\Delta)
                                     -F_\beta(H)+F_\beta(H')\bigr].
 \end{aligned}
\]
Adding the two entropy identities gives the exact formula
\[
 D(\rho\|\rho')+D(\rho'\|\rho)
 =\beta\operatorname{Tr}((\rho'-\rho)\Delta).
\]
The estimates above imply
\[
 \begin{aligned}
 \operatorname{Tr}((\rho'-\rho)\Delta)
 &\le |\operatorname{Tr}(\rho'\Delta)|
      +|\operatorname{Tr}(\rho\Delta)|\\
 &\le a\bigl(U_0(\rho)+U_0(\rho')\bigr)+2c,
 \end{aligned}
\]
which proves \eqref{eq:gibbs-entropy-stability}.  Pinsker's inequality
\cite{Nielsen_Chuang_2010} gives, in both directions,
\[
 \|\rho-\rho'\|_1^2\le2D(\rho\|\rho'),
 \qquad
 \|\rho-\rho'\|_1^2\le2D(\rho'\|\rho).
\]
Taking the smaller right-hand side and using
$2\min\{x,y\}\le x+y$ yields
\[
 \|\rho-\rho'\|_1^2
 \le D(\rho\|\rho')+D(\rho'\|\rho),
\]
and \eqref{eq:gibbs-trace-stability} follows from
\eqref{eq:gibbs-entropy-stability}.
\end{proof}

\smallsection{The Coulomb multiplier and the global cutoff}
The next estimate controls multiplication by the Coulomb potential.
We keep the Hamiltonian \eqref{eq:model-hamiltonian}, the coupling sums
\eqref{eq:coupling-sums}, and the global cutoff
\eqref{eq:global-interaction-definition}.

\begin{lemma}\label{lem:coulomb-multiplier}
Write $x=(x_1,\ldots,x_n)$ and let $dx$ be Lebesgue measure on
$(\mathbb R^d)^n$.  For $d\in\{2,3\}$, each pair $i<j$, and
$\psi\in Q(H_{0,n})$,
\begin{align}
 \|w_d(x_i-x_j)\psi\|^2
 &\le9\langle\psi,(h_i+h_j)\psi\rangle,
                         \label{eq:two-body-operator-bound}\\
 \int_{(\mathbb R^d)^n} |w_d(x_i-x_j)|\,|\psi(x)|^2\,dx
 &\le3\|\psi\|\langle\psi,(h_i+h_j)\psi\rangle^{1/2}.
                         \label{eq:two-body-mixed-bound}
\end{align}
Consequently,
\begin{align}
 \|W_nH_{0,n}^{-1/2}\|&\le3\sqrt{A_nB_n},
                         \label{eq:many-body-operator-bound}\\
 |w_d(x_i-x_j)|&\le\eta(h_i+h_j)+9/(4\eta)\quad(\eta>0),
                         \label{eq:two-body-small-form}\\
 \|(P_M\otimes P_M)W_{12}(P_M\otimes P_M)\|
 &\le3\sqrt{2\lambda_M}\le C_d(1+M)^{1/(2d)}.
                         \label{eq:pair-matrix-bound}
\end{align}
The middle inequality is in the sense of quadratic forms.
\end{lemma}

\begin{proof}
We prove the estimates for smooth compactly supported vectors and then
pass to the limit in the oscillator form norm.  We make the orthogonal change of variables
$r=(x_i-x_j)/\sqrt2$, $s=(x_i+x_j)/\sqrt2$, and denote the transformed
vector by $u$.  We regard the other variables as parameters; $\nabla_r$
is the gradient in the relative coordinate $r$.  For $d=3$, Hardy's inequality gives
\[
 \|w_3(x_i-x_j)\psi\|^2
 =\frac12\||r|^{-1}u\|^2
 \le2\|\nabla_ru\|^2
 \le2\langle\psi,(h_i+h_j)\psi\rangle.
\]
For $d=2$, we define $F(r)=\|u(r,\cdot)\|_{L^2}$, integrating all
variables except $r$.  For a scalar function $F$ on $\mathbb R^2$ and $1\le p<\infty$,
$\|F\|_{L^p}=(\int_{\mathbb R^2}|F|^p)^{1/p}$ is the Lebesgue norm.
Then
$\|F\|_{L^2}=\|u\|$, $\|\nabla F\|_{L^2}\le\|\nabla_ru\|$, and the planar
interpolation inequality $\|F\|_{L^4}^2\le\sqrt2\|F\|_{L^2}\|\nabla F\|_{L^2}$
yields
\[
 \begin{aligned}
 \int_{|r|\le1}|\log|r||^2F(r)^2\,dr
 &\le\left(\int_{|r|\le1}|\log|r||^4\,dr\right)^{1/2}\|F\|_{L^4}^2\\
 &\le\sqrt{3\pi}\,\|u\|\,\|\nabla_ru\|.
 \end{aligned}
\]
Here $\int_{|r|\le1}|\log|r||^4\,dr=3\pi/2$.
On $|r|>1$, $(\log|r|)^2\le |r|^2$.  With
$E_{ij}=\langle\psi,(h_i+h_j)\psi\rangle$ and $h_i+h_j\ge4$, it follows
that
\[
 \begin{aligned}
 \|\log|r|\,u\|^2
 &\le\sqrt{3\pi}\,\|u\|\,\|\nabla_ru\|+\||r|u\|^2
 \le(1+\sqrt{3\pi}/2)E_{ij},\\
 \|\log(\sqrt2|r|)u\|
 &\le\left(\sqrt{1+\sqrt{3\pi}/2}+\frac{\log2}{4}\right)E_{ij}^{1/2}
 <3E_{ij}^{1/2}.
 \end{aligned}
\]
This proves \eqref{eq:two-body-operator-bound}.  Cauchy--Schwarz then gives
\[
 \begin{aligned}
 \int |w_d(x_i-x_j)|\,|\psi|^2
 &=\langle |w_d(x_i-x_j)|\psi,\psi\rangle\\
 &\le \|w_d(x_i-x_j)\psi\|\,\|\psi\|\\
 &\le3\|\psi\|\langle\psi,(h_i+h_j)\psi\rangle^{1/2},
 \end{aligned}
\]
which is \eqref{eq:two-body-mixed-bound}.  Cauchy--Schwarz with weights $|\alpha_{n,i,j}|$ and
\eqref{eq:coupling-sums} give
\[
 \begin{aligned}
 \|W_n\psi\|^2
 &\le B_n\sum_{i<j}|\alpha_{n,i,j}|\,
                         \|w_d(x_i-x_j)\psi\|^2\\
 &\le9B_n\sum_i\left(\sum_{j\ne i}|\alpha_{n,i,j}|\right)
                         \langle\psi,h_i\psi\rangle
 \le9A_nB_n\langle\psi,H_{0,n}\psi\rangle.
 \end{aligned}
\]
If the interaction vanishes, the estimate is immediate.  Otherwise, since
$H_{0,n}\ge d n>0$, replacing $\psi$ by $H_{0,n}^{-1/2}\varphi$ gives
\[
 \|W_nH_{0,n}^{-1/2}\varphi\|^2
 \le9A_nB_n\|\varphi\|^2,
\]
so
\[
 \|W_nH_{0,n}^{-1/2}\|\le3\sqrt{A_nB_n}.
\]
This proves \eqref{eq:many-body-operator-bound}.

Next, \eqref{eq:two-body-mixed-bound} and the scalar inequality
$3ab\le \eta b^2+9a^2/(4\eta)$ give, for every $\eta>0$,
\[
 \begin{aligned}
 \langle\psi,|w_d(x_i-x_j)|\psi\rangle
 &\le3\|\psi\|\langle\psi,(h_i+h_j)\psi\rangle^{1/2}\\
 &\le \eta\langle\psi,(h_i+h_j)\psi\rangle
      +\frac9{4\eta}\|\psi\|^2.
 \end{aligned}
\]
This is the quadratic-form inequality \eqref{eq:two-body-small-form}.

Finally, if $\psi\in\operatorname{Ran}(P_M\otimes P_M)$, then
\[
 \langle\psi,(h_1+h_2)\psi\rangle
 \le2\lambda_M\|\psi\|^2.
\]
Therefore \eqref{eq:two-body-mixed-bound} gives
\[
 |\langle\psi,W_{12}\psi\rangle|
 \le3\sqrt{2\lambda_M}\,\|\psi\|^2,
\]
and hence
\[
 \|(P_M\otimes P_M)W_{12}(P_M\otimes P_M)\|
 \le3\sqrt{2\lambda_M}.
\]
For each nonnegative integer $j$, the energy $d+2j$ has multiplicity
$\binom{j+d-1}{d-1}$.  Writing $\#$ for the number of elements of a
finite set, we obtain
\begin{equation}
\label{eq:oscillator-counting}
 \#\{k:\lambda_k\le d+2j\}=\binom{j+d}{d},\qquad
 c_dM^{1/d}\le\lambda_M\le C_dM^{1/d},\qquad
 \lambda_{M+1}\ge M^{1/d}.
\end{equation}
This proves the last bound in \eqref{eq:pair-matrix-bound}.  Summing
\eqref{eq:two-body-small-form} shows that $W_n$ and $W_{n,M}$ have
arbitrarily small relative form bound with respect to $H_{0,n}$, for
every fixed set of couplings.
\end{proof}

\begin{lemma}\label{lem:global-coulomb-error}
For $\Delta_{n,M}=W_n-W_{n,M}$ and every $\psi\in Q(H_{0,n})$,
\begin{equation}
\label{eq:explicit-cutoff-error}
 |\Delta_{n,M}[\psi]|
 \le6\sqrt{A_nB_n}\,\lambda_{M+1}^{-1/2}
                           \langle\psi,H_{0,n}\psi\rangle.
\end{equation}
\end{lemma}

\begin{proof}
We write $P=\Pi_M$, $Q=I-P$, and
$E=\langle\psi,H_{0,n}\psi\rangle$.  Since $P$ commutes with $H_{0,n}$ and every vector in
$\operatorname{Ran}Q$ has one-particle energy at least $\lambda_{M+1}$,
\[
 H_{0,n}Q\ge\lambda_{M+1}Q.
\]
Consequently,
\[
 \lambda_{M+1}\|Q\psi\|^2
 \le\langle Q\psi,H_{0,n}Q\psi\rangle
 \le E,
\]
so
\[
 \|Q\psi\|\le\lambda_{M+1}^{-1/2}E^{1/2}.
\]
Also,
\[
 \langle P\psi,H_{0,n}P\psi\rangle\le E.
\]
Because $W_{n,M}=PW_nP$ as a quadratic form,
\[
 \begin{aligned}
 \Delta_{n,M}[\psi]
 &=\langle P\psi+Q\psi,W_n(P\psi+Q\psi)\rangle
   -\langle P\psi,W_nP\psi\rangle\\
 &=\langle Q\psi,W_n\psi\rangle
   +\langle W_nP\psi,Q\psi\rangle.
 \end{aligned}
\]
Lemma~\ref{lem:coulomb-multiplier} makes both terms on the last line
well-defined.  Therefore
\[
 |\Delta_{n,M}[\psi]|
 \le\|Q\psi\|\bigl(\|W_n\psi\|+\|W_nP\psi\|\bigr)
 \le6\sqrt{A_nB_n}\,\lambda_{M+1}^{-1/2}E.
\]
\end{proof}

\smallsection{Thermal energy and approximation bounds}
The preceding form bound reduces the approximation problem to a bound
on the mean oscillator energy.  We need the same bound for $H_n$ and
$H_{n,M}$, uniformly in $M$.  For any inverse temperature $t>0$,
$F_t(K)$ and $\sigma_t(K)$ mean the free energy and Gibbs state defined
above with $\beta$ replaced by $t$.
We use $\coth t=\cosh t/\sinh t$ for the hyperbolic cotangent.

\begin{proposition}\label{prop:thermal-energy}
For $d\in\{2,3\}$, $\beta>0$, and $t>0$, define
\begin{equation}
\label{eq:explicit-thermal-constant}
 \begin{aligned}
 u_d(t)&=F_t(h)=\frac d t\log(2\sinh t),\\
 T_{\beta,d}&=4|u_d(\beta)|+|u_d(\beta/4)|
                         +12\sqrt{2d\coth\beta}+18.
 \end{aligned}
\end{equation}
For $K=H_n$ or $H_{n,M}$, one has
\begin{equation}
\label{eq:thermal-energy-bound}
 \operatorname{Tr}(\sigma_\beta(K)H_{0,n})
 \le T_{\beta,d}\bigl[n+(A_n+1)B_n\bigr].
\end{equation}
If $d=3$ and all couplings are nonnegative, the sharper bound is
\begin{equation}\label{eq:repulsive-thermal-energy}
 \operatorname{Tr}(\sigma_\beta(K)H_{0,n})
 \le n\frac{2d}{\beta}\log(2\cosh(\beta/2))
       +6B_n\sqrt{2d\coth\beta}
 \le T_{\beta,d}(n+B_n).
\end{equation}
\end{proposition}

\begin{proof}
Summing \eqref{eq:two-body-small-form} and choosing
$\eta=(2A_n)^{-1}$ when $A_n>0$ gives
\begin{equation}
\label{eq:hamiltonian-lower-bound}
 K\ge\tfrac12H_{0,n}-C_n,\qquad C_n=\tfrac92 A_nB_n.
\end{equation}
The same argument applies to the truncated interaction because
$[H_{0,n},\Pi_M]=0$.
For $A_n=0$, both interactions vanish and the same inequality holds.
If $\lambda_k(K)$ and $\lambda_k(H_{0,n})$ are the ordered eigenvalues,
\eqref{eq:hamiltonian-lower-bound} gives
\[
 \lambda_k(K)\ge\tfrac12\lambda_k(H_{0,n})-C_n.
\]
Therefore
\[
 \begin{aligned}
 \mathcal Z_t(K)
 &=\sum_ke^{-t\lambda_k(K)}\le e^{tC_n}\sum_ke^{-t\lambda_k(H_{0,n})/2}
 =e^{tC_n}\mathcal Z_{t/2}(H_{0,n}),
 \end{aligned}
\]
and hence
\[
 F_t(K)\ge F_t(H_{0,n}/2)-C_n
 =\frac n2F_{t/2}(h)-C_n.
\]
For the upper bound, we use the free Gibbs state as a trial state:
$\rho_{0,t}=\sigma_t(H_{0,n})$.  Let
\[
 \rho_{0,t}=\sum_k p_k|u_k\rangle\langle u_k|
\]
be a spectral decomposition.  By \eqref{eq:two-body-mixed-bound},
\[
 \begin{aligned}
 \operatorname{Tr}(\rho_{0,t}|W_{ij}|)
 &=\sum_kp_k\langle u_k,|W_{ij}|u_k\rangle \le3\sum_kp_k
     \langle u_k,(h_i+h_j)u_k\rangle^{1/2}.
 \end{aligned}
\]
Cauchy--Schwarz in the probability weights $p_k$ gives
\[
 \operatorname{Tr}(\rho_{0,t}|W_{ij}|)
 \le3\bigl[\operatorname{Tr}(\rho_{0,t}(h_i+h_j))\bigr]^{1/2}.
\]
Since $\rho_{0,t}$ is a product of one-particle oscillator Gibbs states,
\[
 \operatorname{Tr}(\rho_{0,t}h_i)
 =\operatorname{Tr}(\sigma_t(h)h)=d\coth t,
\]
and therefore
\[
 \operatorname{Tr}(\rho_{0,t}|W_{ij}|)
 \le3\sqrt{2d\coth t}.
\]
Moreover $[\Pi_M,\rho_{0,t}]=0$, so
\[
 0\le\Pi_M\rho_{0,t}\Pi_M=\rho_{0,t}\Pi_M\le\rho_{0,t}.
\]
Hence the same estimate holds with $W_{ij}$ replaced by its global
compression.  Thus
\begin{equation}
\label{eq:thermal-free-bracket}
 \frac n2F_{t/2}(h)-C_n\le F_t(K)
 \le nF_t(h)+3B_n\sqrt{2d\coth t}.
\end{equation}
All positive-time heat traces are finite by
\eqref{eq:hamiltonian-lower-bound}.  Set
\[
 U_K(t)=\operatorname{Tr}(\sigma_t(K)K)
       =\frac{\operatorname{Tr}(Ke^{-tK})}{\mathcal Z_t(K)}.
\]
Differentiating the partition function gives
\[
 \partial_t\log\mathcal Z_t(K)
 =-U_K(t).
\]
Differentiating $U_K$ once more yields
\[
 \begin{aligned}
 U_K'(t)
 &=-\operatorname{Tr}(\sigma_t(K)K^2)
   +\operatorname{Tr}(\sigma_t(K)K)^2=-\operatorname{Tr}\bigl(\sigma_t(K)(K-U_K(t))^2\bigr)
 \le0.
 \end{aligned}
\]
Thus $U_K$ is decreasing.  For $s\in[t/2,t]$,
$U_K(t)\le U_K(s)$, and therefore
\[
 \begin{aligned}
 U_K(t)
 &\le\frac2t\int_{t/2}^tU_K(s)\,ds=-\frac2t\int_{t/2}^t\partial_s\log\mathcal Z_s(K)\,ds\\
 &=\frac2t\bigl(\log\mathcal Z_{t/2}(K)-\log\mathcal Z_t(K)\bigr)=2F_t(K)-F_{t/2}(K).
 \end{aligned}
\]
Combining this inequality with \eqref{eq:hamiltonian-lower-bound} and
\eqref{eq:thermal-free-bracket} yields
\[
 \begin{aligned}
 \operatorname{Tr}(\sigma_\beta(K)H_{0,n})
 &\le2U_K(\beta)+2C_n\\
 &\le4nF_\beta(h)-nF_{\beta/4}(h)
           +12B_n\sqrt{2d\coth\beta}+18A_nB_n.
 \end{aligned}
\]
Taking absolute values of the two free-energy terms and using
\eqref{eq:explicit-thermal-constant} proves
\eqref{eq:thermal-energy-bound}.
For repulsive three-dimensional interactions, $K\ge H_{0,n}$, also
for the compression.  Replace the lower free-energy bound by
$F_t(K)\ge n u_d(t)$ and use the same trial upper bound.  Then
\[
 \operatorname{Tr}(\sigma_\beta(K)H_{0,n})\le U_K(\beta)
 \le n[2u_d(\beta)-u_d(\beta/2)]+6B_n\sqrt{2d\coth\beta}.
\]
The bracket equals $(2d/\beta)\log(2\cosh(\beta/2))$.
It is bounded by $T_{\beta,d}$: $u_d$ is increasing, so
$u_d(\beta/2)\ge u_d(\beta/4)$, and the definition of $T_{\beta,d}$
dominates $2|u_d(\beta)|+|u_d(\beta/4)|$.
This proves \eqref{eq:repulsive-thermal-energy}.
\end{proof}

\subsection{Relative entropy and trace-distance convergence}
The same stability estimate controls the Gibbs state itself.  We now
combine the free-energy estimate with the symmetrized relative-entropy
bound from Lemma~\ref{lem:gibbs-stability}.

\begin{proof}[Proofs of Theorem~\ref{thm:coulomb-approximation}]
We combine Lemmas~\ref{lem:gibbs-stability} and~\ref{lem:global-coulomb-error}
with Proposition~\ref{prop:thermal-energy}.  We take $H=H_n$, $H'=H_{n,M}$,
$c=0$, and $a=6\sqrt{A_nB_n}\lambda_{M+1}^{-1/2}$, using
\eqref{eq:explicit-cutoff-error}.  Equations
\eqref{eq:gibbs-free-stability}, \eqref{eq:gibbs-trace-stability},
\eqref{eq:gibbs-entropy-stability}, and \eqref{eq:thermal-energy-bound} give
\[
 \begin{aligned}
 |F_\beta(H_n)-F_\beta(H_{n,M})|
 &\le6T_{\beta,d}\Xi_n\lambda_{M+1}^{-1/2},\\
 \|\sigma_\beta(H_n)-\sigma_\beta(H_{n,M})\|_1^2
 &\le12\beta T_{\beta,d}\Xi_n\lambda_{M+1}^{-1/2},\\
 D(\sigma_\beta(H_n)\|\sigma_\beta(H_{n,M}))
 +D(\sigma_\beta(H_{n,M})\|\sigma_\beta(H_n))
 &\le12\beta T_{\beta,d}\Xi_n\lambda_{M+1}^{-1/2}.
 \end{aligned}
\]
By \eqref{eq:oscillator-counting}, \eqref{eq:main-free-energy-error} and
\eqref{eq:main-gibbs-error} follow with
\begin{equation}
\label{eq:explicit-approximation-constants}
 C_F=6T_{\beta,d},\qquad C_\rho=\sqrt{12\beta T_{\beta,d}}.
\end{equation}
To bound the coefficient $\Xi_n$ in \eqref{eq:coupling-envelope} by
the largest pair coupling, we take
$a=\alpha_n^{\max}$ and use $A_n\le na$, $B_n\le n^2a/2$.
Consequently,
\[
 \begin{aligned}
 \sqrt{A_nB_n}
 &\le \frac{n^{3/2}a}{\sqrt2},\\
 n+(A_n+1)B_n
 &\le n+\frac{n^3a^2}{2}+\frac{n^2a}{2},
 \end{aligned}
\]
and hence
\[
 \Xi_n
 \le\frac{n^{5/2}a}{\sqrt2}
   +\frac{n^{7/2}a^2+n^{9/2}a^3}{2\sqrt2}.
\]
The elementary inequality
\[
 2(na)^2\le na+(na)^3
\]
absorbs the middle term into the first and third terms.  Therefore
\[
 \Xi_n\le n^{5/2}a+n^{9/2}a^3=Q_n(a).
\]
The zero-coupling case follows directly.  Using
\eqref{eq:repulsive-thermal-energy} in the same argument proves the
claims with $\Xi_n^+$.  In particular, for dense repulsive couplings
bounded by $a$, the coefficient is
$\mathcal O(n^{5/2}a+n^{7/2}a^2)$.
\end{proof}

\section{Spectral gap analysis}
\label{section:Pert_Theory}

\subsection{Quadratic Hamiltonians}
\smallsection{The sampler}
\label{sec:sampler-definition}

We describe the dynamics on trace-class operators and use its
self-adjoint realization on Hilbert--Schmidt operators to estimate the
gap, using the Hilbert--Schmidt inner product defined in the introduction.
Throughout this section, $H=H_{0,n}+W$ with $W=W^*$ bounded.  We set
$\widetilde H=H+\|W\|+1\ge1$.
The domains $D(H)=D(H_{0,n})$ and $Q(H)=Q(H_{0,n})$ have equivalent
norms.  For every oscillator ladder $A\in\mathcal A$,
\begin{equation}
\label{eq:generation-domain-bounds}
 \|A\widetilde H^{-1/2}\|<\infty,\qquad
 \|\widetilde H^{1/2}A\widetilde H^{-1}\|<\infty.
\end{equation}
The ladder identities imply
\[
 \|A(H_{0,n}+1)^{-1/2}\|<\infty,
 \qquad
 \|(H_{0,n}+1)^{1/2}A(H_{0,n}+1)^{-1}\|<\infty.
\]
Since
\[
 H_{0,n}+1=\widetilde H-(W+\|W\|),
\]
we have
\[
 (H_{0,n}+1)\widetilde H^{-1}
 =I-(W+\|W\|)\widetilde H^{-1},
\]
which is bounded.  Likewise, bounded perturbation theory gives constants
$c_1,c_2>0$ such that
\[
 c_1(H_{0,n}+1)\le\widetilde H\le c_2(H_{0,n}+1)
\]
in the sense of quadratic forms.  Combining these estimates gives
\eqref{eq:generation-domain-bounds}.
Eigenvalue comparison also gives
\begin{equation}
\label{eq:bounded-perturbation-heat-trace}
 \operatorname{Tr}e^{-tH}\le e^{t\|W\|}
                  \operatorname{Tr}e^{-tH_{0,n}}<\infty\qquad(t>0).
\end{equation}
These estimates verify the domain and heat-trace assumptions in
\cite[Condition~A and the Gibbs Dirichlet-form construction]{BeckerRouzeSalzmannToAppearcmp}.

We set
\[
 \rho_H=\sigma_\beta(H),
 \qquad
 g_H=\rho_H^{1/2}.
\]
For $A\in\mathcal A$ and $\psi\in Q(H)$, the filtered jump is
\[
 L_H^A\psi
 =\int_{\mathbb R}f(t)e^{itH}Ae^{-itH}\psi\,dt.
\]
For a self-adjoint operator $T$ and a Borel set $B\subset\mathbb R$,
$\mathbf1_B(T)$ denotes its spectral projection on $B$.
For each eigenvalue $E$ of $H$, we put
$P_E(H)=\mathbf1_{\{E\}}(H)$.  A Bohr frequency is an energy difference
$\nu=E'-E$ between two eigenvalues of $H$.  We define the corresponding part of a jump on the spectral core by
\[
 A_\nu
 =\sum_{E'-E=\nu}P_{E'}(H)AP_E(H).
\]
For $t\in\mathbb R$ and a Hilbert--Schmidt matrix $X$ supported in a
finite $H$-spectral subspace, we define the defect
\begin{equation}
\label{eq:smooth-defect}
 D_H^A(t;X)=\sum_\nu e^{it\nu}\widehat f(\nu)
       \bigl(A_\nu X-e^{\beta\nu/2}XA_\nu\bigr).
\end{equation}
The sum runs over the Bohr frequencies of $H$.
We also define the filter moments, the response bound $h_f$, and the
integration weight by
\begin{equation}
\label{eq:filter-moments}
 \begin{aligned}
 m_j(f)&=\int_{\mathbb R}|t|^j|f(t)|\,dt\quad(j=0,1,2),\\
 h_f(t)&=m_0(f)|t|+m_1(f),\qquad
 w_\beta(t)=\frac1{\beta\cosh(2\pi t/\beta)},\\
 I_{\beta,f}&=\int_{\mathbb R}w_\beta(t)h_f(t)^2\,dt.
 \end{aligned}
\end{equation}
The coherent sampler has the closed quadratic form
\begin{equation}
\label{eq:smooth-dirichlet-form}
 \mathcal E_H(X)=-\langle X,L_{\infty,H}X\rangle_2
 =\sum_{A\in\mathcal A}\int_{\mathbb R}w_\beta(t)
       \|D_H^A(t;X)\|_2^2\,dt.
\end{equation}
It is defined by closure from the finite $H$-spectral matrices and
$g_H$.  The middle expression denotes the form value when $X$ lies
only in the form domain.  The core identity follows from
\cite[definition of the KMS correction operator]{BeckerRouzeSalzmannToAppearcmp}.  Its normalization
uses the Fourier identity
\[
 \int_{\mathbb R}w_\beta(t)e^{it\omega}\,dt
 =\frac1{2\cosh(\beta\omega/4)},
 \qquad \omega\in\mathbb R.
\]
The Gibbs vector satisfies
\begin{equation}
\label{eq:smooth-ground-defect}
 A_\nu g_H=e^{\beta\nu/2}g_HA_\nu,\qquad \mathcal E_H(g_H)=0.
\end{equation}

For $0<\sigma_E<\infty$, we define the loss kernel and the Gaussian
averaging density, for $t,s\in\mathbb R$, by
\[
 g_{\sigma_E}(t)=\frac1{2\pi}\int_{\mathbb R}
 \frac{e^{-\nu^2/(8\sigma_E^2)}}{1+e^{\beta\nu/2}}e^{-it\nu}\,d\nu,
 \qquad
 p_{\sigma_E}(s)=\sqrt{\frac2\pi}\sigma_E e^{-2\sigma_E^2s^2}.
\]
We write $X_s^A=e^{isH}L_H^Ae^{-isH}$ for the time-evolved filtered jump.
To specify trace-class products, for $r,s\ge0$ we define
\[
 \begin{aligned}
 D(\mathcal W_H^{r,s})&=\{\widetilde H^{-r}Y\widetilde H^{-s}:
                                         Y\in\mathscr T_1\},\\
 \mathcal W_H^{r,s}(\rho)&=\widetilde H^r\rho\widetilde H^s,
 \qquad
 \|\rho\|_{\mathcal W_H^{r,s}}=\|\mathcal W_H^{r,s}(\rho)\|_1.
 \end{aligned}
\]
A centered dot denotes the trace-class extension of an operator product.
For example, for $\rho\in D(\mathcal W_H^{1/2,1/2})$, we define
\[
 X_s^A\cdot\rho\cdot(X_s^A)^\dagger
 =(X_s^A\widetilde H^{-1/2})\mathcal W_H^{1/2,1/2}(\rho)
                          (X_s^A\widetilde H^{-1/2})^\dagger.
\]
The loss operator and gain map are
\begin{align}
 G_H\psi&=-\sum_{A\in\mathcal A}\int_{\mathbb R}
 g_{\sigma_E}(t)e^{itH}(L_H^A)^\dagger L_H^Ae^{-itH}\psi\,dt,
 \quad\psi\in D(H),\label{eq:loss-integral}\\
 \Phi_H(\rho)&=\sum_{A\in\mathcal A}\int_{\mathbb R}
 p_{\sigma_E}(s)X_s^A\cdot\rho\cdot(X_s^A)^\dagger\,ds.
 \label{eq:gain-integral}
\end{align}
The parameter $\sigma_E$ is suppressed in $G_H$ and $\Phi_H$.
The products $G_H\cdot\rho$ and $\rho\cdot G_H^\dagger$ are likewise
defined by the bounded factors $G_H\widetilde H^{-1}$ and
$\widetilde H^{-1}G_H^\dagger$ on the weighted domains below.

\begin{proposition}\label{prop:sampler-generation}
The closed form \eqref{eq:smooth-dirichlet-form} of the coherent sampler
generates
a KMS-symmetric quantum Markov semigroup.  For $0<\sigma_E<\infty$,
the operator
\[
 \mathcal L_{\sigma_E,H}(\rho)=G_H\cdot\rho+\rho\cdot G_H^\dagger
                                  +\Phi_H(\rho)
\]
extends to the generator of a strongly continuous semigroup of
completely positive trace-preserving maps on $\mathscr T_1$.
Its domain contains
$D(\mathcal W_H^{1/2,1/2})\cap D(\mathcal W_H^{1,0})
\cap D(\mathcal W_H^{0,1})$.
For $W=\Pi_MW\Pi_M$ and $0<\sigma_E\le\infty$, the unique invariant
normal state is $\sigma_\beta(H)$.
\end{proposition}

\begin{proof}
Equations \eqref{eq:generation-domain-bounds} and
\eqref{eq:bounded-perturbation-heat-trace} verify the hypotheses of
\cite[Dirichlet-form construction, trace-class generation, and Gaussian-averaged generation]{BeckerRouzeSalzmannToAppearcmp}.
These results give the coherent construction, its trace-class extension,
and the stated domain for the averaged generator.
The kernel and uniqueness of the invariant normal state are proved in
Theorem~\ref{thm:finite-rank-kernel} and the argument following it.
\end{proof}

For every $0<\sigma_E\le\infty$, we write
$\mathcal E_{\sigma_E,H}$ for the closed quadratic form of
$-L_{\sigma_E,H}$; in particular $\mathcal E_{\infty,H}=\mathcal E_H$.
We define the embedding
$\iota_{2,H}(X)=\rho_H^{1/4}X\rho_H^{1/4}$ from $\mathscr T_2$ to
$\mathscr T_1$.  For $X\in\mathscr T_2$ and $t\ge0$, the semigroups satisfy
\begin{equation}
\label{eq:sampler-intertwining}
 e^{t\mathcal L_{\sigma_E,H}}\iota_{2,H}(X)
 =\iota_{2,H}(e^{tL_{\sigma_E,H}}X).
\end{equation}
Schatten H\"older shows that the map in
\eqref{eq:sampler-intertwining} has norm at most one.  For an initial
density operator $\rho_{\rm ini}\le\mathfrak c\rho_H$, with
$\mathfrak c\ge1$, and a gap lower bound $\gamma>0$, we obtain the
following estimate.  For $0<\varepsilon<1$, the mixing time
$t_{\rm mix}(\varepsilon)$ is the least time after which this initial
state stays within trace distance $\varepsilon$ of $\rho_H$.  Then
\begin{equation}
\label{eq:warm-start-mixing}
 \|e^{t\mathcal L_{\sigma_E,H}}(\rho_{\rm ini})-\rho_H\|_1
 \le\mathfrak c\,e^{-\gamma t},\qquad
 t_{\rm mix}(\varepsilon)\le\gamma^{-1}\log(\mathfrak c/\varepsilon).
\end{equation}
The domination assumption gives the required Hilbert--Schmidt
representation.  Define
\[
 Y=\rho_H^{-1/2}\rho_{\rm ini}\rho_H^{-1/2}
\]
on the support of $\rho_H$.  Then
\[
 0\le Y\le\mathfrak c I,
 \qquad
 \rho_{\rm ini}=\rho_H^{1/2}Y\rho_H^{1/2}.
\]
Set
\[
 X=\rho_H^{1/4}Y\rho_H^{1/4}.
\]
Schatten H\"older gives
\[
 \|X\|_2
 \le\|Y\|\,\|\rho_H^{1/2}\|_2
 \le\mathfrak c.
\]
Moreover,
\[
 \langle g_H,X\rangle_2
 =\operatorname{Tr}(\rho_H^{1/2}Y\rho_H^{1/2})
 =\operatorname{Tr}\rho_{\rm ini}=1,
\]
so $X-g_H\perp g_H$.  Since $\|g_H\|_2=1$,
\[
 \|X-g_H\|_2^2
 =\|X\|_2^2-1
 \le\mathfrak c^2-1
 \le\mathfrak c^2.
\]
The spectral-gap estimate on $g_H^\perp$ therefore gives
\[
 \|e^{tL_{\sigma_E,H}}(X-g_H)\|_2
 \le\mathfrak c e^{-\gamma t}.
\]
Using the intertwining relation and contractivity of $\iota_{2,H}$,
\[
 \|e^{t\mathcal L_{\sigma_E,H}}(\rho_{\rm ini})-\rho_H\|_1
 \le\mathfrak c e^{-\gamma t},
\]
which proves \eqref{eq:warm-start-mixing}.

We define the coherent sampler through the closed form
\eqref{eq:smooth-dirichlet-form}.  The averaging identity
\eqref{eq:smooth-gaussian-average} shows that its gap bounds also hold
for every finite $\sigma_E>0$.

\smallsection{The free oscillator sampler}
The free sampler can be diagonalized explicitly.  We record both its
gap and its form domain, since the perturbation argument uses both.
On one particle, $a_k=(x_k+\partial_{x_k})/\sqrt2$, $1\le k\le d$,
is the lowering operator.  We enumerate all modes by
$r=d(i-1)+k$, and write $a_r=a_{ik}$ and
$n_{ik}=a_{ik}^\dagger a_{ik}=N_r$ for their occupation operators.
The total occupation operator is
\[
 N_{\rm tot}=\sum_{i=1}^n\sum_{k=1}^d n_{ik}.
\]
The oscillator identities give
\[
 h=-\Delta+|x|^2
   =d+2\sum_{k=1}^d a_k^\dagger a_k,
\]
and therefore
\[
 H_{0,n}=\sum_{i=1}^nh_i=dn+2N_{\rm tot}.
\]
We set
\[
 \rho_0=\sigma_\beta(H_{0,n}),
 \qquad
 g_0=\rho_0^{1/2}.
\]
With the gap from \eqref{eq:oscillator-gap-definition}, define
\[
 \nu_-=|\widehat f(-2)|^2,
 \qquad
 \nu_+=|\widehat f(2)|^2=e^{-2\beta}\nu_-,
\]
and
\[
 \lambda:=\lambda_{\rm OU}
 =\frac{\nu_--\nu_+}{2}>0.
\]
We abbreviate $L_{H_{0,n}}=L_{\infty,H_{0,n}}$ and
$L_h=L_{\infty,h}$ for the one-particle generator.  Each free jump changes the energy by a single amount, either $-2$ or
$2$.  Gaussian averaging therefore leaves the free generator unchanged.
It is the sum of $dn$ commuting one-mode generators:
\begin{equation}
\label{eq:reference-generator}
\begin{aligned}
 L_{H_{0,n}}X=\sum_{i,k}\Big[&
 -\frac{\nu_-+\nu_+}{2}(n_{ik}X+Xn_{ik})-\nu_+X\\
 &+\sqrt{\nu_-\nu_+}
 (a_{ik}Xa_{ik}^\dagger+a_{ik}^\dagger Xa_{ik})\Big].
\end{aligned}
\end{equation}
For one mode, write $a=(x+\partial_x)/\sqrt2$ and let $|m\rangle$,
$m\in\mathbb N_0$, be the normalized Hermite vectors, with
$a^\dagger a|m\rangle=m|m\rangle$.  We choose their phases so that
$a|0\rangle=0$ and $a|m\rangle=\sqrt m\,|m-1\rangle$ for $m\ge1$.
For vectors $u,v$, the rank-one
operator $|u\rangle\langle v|$ maps $z$ to $u\langle v,z\rangle$.
The Hilbert--Schmidt basis $|m\rangle\langle k|$ identifies this space
with $\ell^2(\mathbb N_0)\otimes\ell^2(\mathbb N_0)$, where $\ell^2$
is the space of square-summable sequences.  We put
$b_1X=aX$ and $b_2X=Xa^\dagger$.  Their adjoints are $b_1^\dagger X=a^\dagger X$
and $b_2^\dagger X=Xa$.  If $\Delta_\nu=\nu_--\nu_+$, the operators
\[
 c_1=\frac{\sqrt{\nu_-}b_1-\sqrt{\nu_+}b_2^\dagger}{\sqrt{\Delta_\nu}},
 \qquad
 c_2=\frac{\sqrt{\nu_-}b_2-\sqrt{\nu_+}b_1^\dagger}{\sqrt{\Delta_\nu}}
\]
satisfy the canonical commutation relations.  We write
$L_{-\partial_x^2+x^2}$ for the one-mode coherent reference generator.
Then
\begin{equation}
\label{eq:reference-diagonalization}
 -L_{-\partial_x^2+x^2}=\frac{\Delta_\nu}{2}
              (c_1^\dagger c_1+c_2^\dagger c_2)
\end{equation}
The two-mode squeeze with parameter $r_{\rm sq}>0$, defined by
$\tanh r_{\rm sq}=\sqrt{\nu_+/\nu_-}=e^{-\beta}$, implements this change of
operators unitarily.  Thus \eqref{eq:reference-diagonalization} has
spectrum $\{\lambda j:j\in\mathbb N_0\}$, finite multiplicities, and
one-dimensional kernel generated by
$(1-e^{-2\beta})^{1/2}\sum_{m\ge0}e^{-\beta m}|m\rangle\langle m|$;
this is also the oscillator spectral calculation in \cite{OU}.
Finite tensor products prove that
$\mathsf A_{0,n}:=-L_{H_{0,n}}$ has compact resolvent, kernel
$\mathbb Cg_0$, and gap $\lambda$.
For $X\in D(\mathsf A_{0,n}^{1/2})$, its closed form satisfies
\begin{equation}
\label{eq:smooth-free-poincare}
 \mathcal E_0(X):=\|\mathsf A_{0,n}^{1/2}X\|_2^2
 \ge\lambda\big(\|X\|_2^2-|\langle g_0,X\rangle_2|^2\big).
\end{equation}
To identify the form domain, we define
\[
 \mathcal N X=N_{\rm tot}X+XN_{\rm tot},\qquad
 c_-=\frac{(\sqrt{\nu_-}-\sqrt{\nu_+})^2}{2}>0,
 \qquad c_+=\frac{(\sqrt{\nu_-}+\sqrt{\nu_+})^2}{2}.
\]
For one mode, Cauchy--Schwarz and
$2|\langle u,v\rangle|\le\|u\|^2+\|v\|^2$ give
\[
 \begin{aligned}
 |\langle X,b_1b_2X\rangle_2|
 &\le\|b_1^\dagger X\|_2\,\|b_2X\|_2\\
 &\le\frac12\bigl(\|b_1^\dagger X\|_2^2+\|b_2X\|_2^2\bigr),
 \end{aligned}
\]
and similarly for $b_1^\dagger b_2^\dagger$.  Using the canonical
commutation relations to move creation operators to the right yields
\[
 |\langle X,(b_1b_2+b_1^\dagger b_2^\dagger)X\rangle_2|
 \le\langle X,(b_1^\dagger b_1+b_2^\dagger b_2+1)X\rangle_2.
\]
Summing this bound in \eqref{eq:reference-generator} gives
\begin{equation}
\label{eq:reference-form-domain}
 \begin{aligned}
 c_-\langle X,\mathcal NX\rangle_2-Cdn\|X\|_2^2
 &\le\mathcal E_0(X)
 \le c_+\langle X,\mathcal NX\rangle_2+Cdn\|X\|_2^2,\\
 D(\mathsf A_{0,n}^{1/2})&=D(\mathcal N^{1/2})
 =\{X:N_{\rm tot}^{1/2}X,\ XN_{\rm tot}^{1/2}\in\mathscr T_2\}.
 \end{aligned}
\end{equation}
The constant $C$ depends only on the filter and $\beta$.
Matrices with finite support in the oscillator basis therefore form
a core for this quadratic form.
The one-particle gap on $\mathscr T_2(L^2(\mathbb R^d))$ will be used
in Lemma~\ref{lem:global-projection-resolvent}.

\subsection{Finite-rank perturbative stability}
\label{sec:bounded-commutator-gap}

We write each quadratic form as the squared norm of a closed operator.
We first bound the difference of these operators, then apply the
min--max principle.  This avoids comparing the free and perturbed Gibbs
vectors directly.

We consider $H=H_{0,n}+W$, where $W=W^*$ is bounded and
\begin{equation}
\label{eq:bounded-commutator-data}
 b:=\|W\|,\qquad
 C_A:=\|[W,A]\|<\infty\quad(A\in\mathcal A),\qquad
 C_W:=\left(\sum_{A\in\mathcal A}C_A^2\right)^{1/2}.
\end{equation}
Set $T_A=[W,A]$ and use the joint row and column bound
\begin{equation}\label{eq:joint-commutator-norm}
 C_*(W):=\max\left\{
  \left\|\sum_A T_A^*T_A\right\|^{1/2},
  \left\|\sum_A T_AT_A^*\right\|^{1/2}\right\}\le C_W.
\end{equation}
The two norms agree for the adjoint-pair set $\mathcal A$.
We first define the commutators weakly on vectors $\phi,\psi$ with
finite oscillator spectral support by
\[
 \langle\phi,[W,A]\psi\rangle
 =\langle W\phi,A\psi\rangle
  -\langle A^\dagger\phi,W\psi\rangle.
\]
By assumption, this form extends to a bounded operator.  The definition
of the adjoint
domain and approximation in the graph norm of $A$ then give
$W D(A)\subset D(A)$ and $WA-AW=[W,A]$ on $D(A)$.
These assumptions hold for both classes used here: interactions
supported in a finite oscillator spectral subspace, and bounded
multiplication operators whose coefficients belong to $W^{1,\infty}$,
the space of essentially bounded functions with essentially bounded
weak first derivatives.
The domain and generation hypotheses have already been verified in
\eqref{eq:generation-domain-bounds} and
\eqref{eq:bounded-perturbation-heat-trace}.

The response operators are the time-evolved filtered jumps.  For
$A\in\mathcal A$ and $t\in\mathbb R$, we define
\begin{equation}
\label{eq:smooth-responses}
 \begin{aligned}
 L_{H,f}^A(t)&=\int_{\mathbb R}f(s)e^{i(t+s)H}Ae^{-i(t+s)H}\,ds,\\
 R_{H,f}^A(t)&=\int_{\mathbb R}\overline{f(s)}e^{i(t+s)H}Ae^{-i(t+s)H}\,ds.
 \end{aligned}
\end{equation}
The free responses are $L_{0,f}^A(t)=L_{H_{0,n},f}^A(t)$ and
$R_{0,f}^A(t)=R_{H_{0,n},f}^A(t)$.
The KMS relation \eqref{eq:KMSFilterFunction} gives
\[
 e^{\beta\nu/2}\widehat f(\nu)
 =\widehat{\overline f}(\nu),
\]
where $\widehat{\overline f}$ is the Fourier transform of the complex
conjugate of $f$.  Resolving the two response operators into Bohr
frequencies therefore gives
\[
 L_{H,f}^A(t)X
 =\sum_\nu e^{it\nu}\widehat f(\nu)A_\nu X,
\]
and
\[
 XR_{H,f}^A(t)
 =\sum_\nu e^{it\nu}e^{\beta\nu/2}\widehat f(\nu)XA_\nu.
\]
Subtracting the two expressions shows that the defect from
\eqref{eq:smooth-defect} satisfies
\begin{equation}
 D_H^A(t;X)=L_{H,f}^A(t)X-XR_{H,f}^A(t).
 \label{eq:smooth-response-defect}
\end{equation}
We use the Hilbert space
\[
 \mathfrak K_n
 =\bigoplus_{A\in\mathcal A}
 L^2(\mathbb R,w_\beta(t)\,dt;\mathscr T_2(\mathcal H_n)).
\]
Thus an element $Y=(Y_A)_{A\in\mathcal A}$ has norm
\[
 \|Y\|_{\mathfrak K_n}^2
 =\sum_{A\in\mathcal A}\int_{\mathbb R}
 w_\beta(t)\|Y_A(t)\|_2^2\,dt.
\]
We collect the defects into the closed operator $\mathfrak D_H$, called
the defect column, by
\[
 (\mathfrak D_HX)_A(t)=D_H^A(t;X).
\]
Then
\[
 \|\mathfrak D_HX\|_{\mathfrak K_n}^2
 =\mathcal E_H(X).
\]
Its free counterpart is $\mathfrak D_0=\mathfrak D_{H_{0,n}}$.
Both are defined by closure from the spectral core used for
\eqref{eq:smooth-dirichlet-form}.

\begin{lemma}
\label{lem:smooth-common-core}
Under \eqref{eq:bounded-commutator-data}, there is a bounded operator
$K_W:\mathscr T_2(\mathcal H_n)\to\mathfrak K_n$ such that
\begin{equation}
\label{eq:smooth-column-identity}
 \mathfrak D_H=\mathfrak D_0+K_W,
 \qquad D(\mathfrak D_H)=D(\mathfrak D_0),
 \qquad \|K_W\|\le2\sqrt{I_{\beta,f}}\,C_*(W).
\end{equation}
Consequently,
\begin{equation}
\label{eq:smooth-form-triangle}
 \big|\mathcal E_H(X)^{1/2}-\mathcal E_0(X)^{1/2}\big|
 \le2\sqrt{I_{\beta,f}}\,C_*(W)\|X\|_2
 \qquad(X\in D(\mathfrak D_0)).
\end{equation}
\end{lemma}

\begin{proof}
For a lowering jump we set $\theta_{a_{ik}}=-2$, and for a raising
jump $\theta_{a_{ik}^\dagger}=2$.  For every real time $u$,
$e^{iuH_{0,n}}Ae^{-iuH_{0,n}}=e^{i\theta_Au}A$.
Differentiation of matrix elements on $D(H_{0,n})$ gives
\begin{equation}
\label{eq:smooth-duhamel}
 e^{iuH}Ae^{-iuH}-e^{iuH_{0,n}}Ae^{-iuH_{0,n}}
 =i\int_0^u e^{i(u-r)H}
       [W,e^{irH_{0,n}}Ae^{-irH_{0,n}}]e^{-i(u-r)H}\,dr.
\end{equation}
Since
\[
 [W,e^{irH_{0,n}}Ae^{-irH_{0,n}}]
 =e^{i\theta_Ar}[W,A],
\]
the integrand in \eqref{eq:smooth-duhamel} has norm at most $C_A$.
Therefore
\[
 \|e^{iuH}Ae^{-iuH}-e^{iuH_{0,n}}Ae^{-iuH_{0,n}}\|
 \le\int_0^{|u|}C_A\,dr=|u|C_A.
\]
Thus the right-hand side extends the difference to a bounded operator.
With $h_f$ defined in \eqref{eq:filter-moments}, the response difference
satisfies
\[
 \begin{aligned}
 \|L_{H,f}^A(t)-L_{0,f}^A(t)\|
 &\le\int_{\mathbb R}|f(s)|\,C_A|t+s|\,ds\\
 &\le C_A\bigl(m_0(f)|t|+m_1(f)\bigr)=C_Ah_f(t),
 \end{aligned}
\]
and the same argument with $\overline f$ gives the corresponding bound
for the right response.  Hence
\begin{equation}
\label{eq:smooth-response-bound}
 \|L_{H,f}^A(t)-L_{0,f}^A(t)\|,
 \ \|R_{H,f}^A(t)-R_{0,f}^A(t)\|
 \le C_Ah_f(t).
\end{equation}
On the spectral core we define $K_W$ componentwise by
\[
 (K_WX)_A(t)=\bigl(L_{H,f}^A(t)-L_{0,f}^A(t)\bigr)X
             -X\bigl(R_{H,f}^A(t)-R_{0,f}^A(t)\bigr).
\]
For the family indexed by $A$, multiplication by
$e^{i\theta_A r}$ and common unitary conjugation preserve both the row
and column norms in \eqref{eq:joint-commutator-norm}.  Minkowski's
inequality in \eqref{eq:smooth-duhamel}, then in the filter integral,
therefore bounds both norms of each response-difference family by
$h_f(t)C_*(W)$.  Left multiplication uses its column norm and right
multiplication its row norm.  Thus
\[
 \|K_WX\|_{\mathfrak K_n}^2
 \le4C_*(W)^2\int_{\mathbb R}w_\beta(t)h_f(t)^2\,dt\,\|X\|_2^2
 =4I_{\beta,f}C_*(W)^2\|X\|_2^2.
\]

It remains to prove equality of the closed domains.  For
$\phi,\psi\in Q(H_{0,n})$, we use the form norm
$\|\phi\|_Q^2=\|\phi\|^2+\langle\phi,H_{0,n}\phi\rangle$.
For a constant $c_{\beta,f}>0$ depending only on $\beta$ and $f$, the
free ladder formula gives
\[
 \|\mathfrak D_0(|\phi\rangle\langle\psi|)\|_{\mathfrak K_n}
 \le c_{\beta,f}\big(\|\phi\|_Q\|\psi\|+\|\phi\|\|\psi\|_Q\big).
\]
Equation \eqref{eq:reference-form-domain} gives the oscillator matrix core.
For a real energy threshold $R$, set $P_R=\mathbf1_{(-\infty,R]}(H)$.
Because $H$ and $H_{0,n}$ have the same form domain with equivalent form
norms,
\[
 \|(I-P_R)\phi\|_Q\longrightarrow0
 \qquad(R\to\infty)
\]
for every $\phi\in Q(H_{0,n})$.  Hence, for rank-one operators,
\[
 \begin{aligned}
 &\|\mathfrak D_0(|\phi\rangle\langle\psi|
       -|P_R\phi\rangle\langle P_R\psi|)\|_{\mathfrak K_n}\\
 &\quad\le c_{\beta,f}\Bigl(
     \|\phi-P_R\phi\|_Q\|\psi\|
    +\|P_R\phi\|_Q\|\psi-P_R\psi\|\\
 &\hspace{43mm}
    +\|\phi-P_R\phi\|\|\psi\|_Q
    +\|P_R\phi\|\|\psi-P_R\psi\|_Q\Bigr)
 \longrightarrow0.
 \end{aligned}
\]
Thus matrices with finite $H$-spectral support form a core for
$\mathfrak D_0$.

Moreover, since $H=H_{0,n}+W$ and $\|W\|=b$,
\[
 H_{0,n}\le H+b.
\]
Therefore
\[
 \operatorname{Tr}(\rho_HH_{0,n})
 \le\operatorname{Tr}(\rho_HH)+b
 <\infty.
\]
Using \eqref{eq:reference-form-domain}, this implies
$g_H\in D(\mathfrak D_0)$.
The finite $H$-spectral projections of $g_H$ converge in this graph
norm because its Gibbs energy is finite.  The form in \eqref{eq:smooth-dirichlet-form} is defined by closure
from the finite $H$-spectral matrices and $g_H$.
On this common core, the operator identity follows from
\eqref{eq:smooth-responses}; at $g_H$ it follows from
\eqref{eq:smooth-ground-defect}.  Since $K_W$ is bounded,
$\mathfrak D_0+K_W$ is closed on $D(\mathfrak D_0)$.
Taking closures proves \eqref{eq:smooth-column-identity}.  The triangle
inequality then gives \eqref{eq:smooth-form-triangle}.
\end{proof}

The next elementary comparison applies to all ordered eigenvalues,
including zero eigenvalues.

\begin{lemma}\label{lem:column-minmax}
Let $D_0,D_1:\mathcal U\to\mathcal V$ be closed densely defined operators
between Hilbert spaces, with the same domain.  Assume that $D(D_0)$,
equipped with its graph norm, embeds compactly into $\mathcal U$.  Suppose
$D_1-D_0$ extends to a bounded operator of norm at most $\kappa\ge0$.
For $i\in\{0,1\}$, the operators $D_i^*D_i$ have compact resolvent.  Their ordered eigenvalues
$\mu_j^{(i)}$, including multiplicities and zero eigenvalues, satisfy
\begin{equation}
\label{eq:column-minmax}
 \big|\sqrt{\mu_j^{(1)}}-\sqrt{\mu_j^{(0)}}\big|\le\kappa
 \qquad(j\ge1).
\end{equation}
\end{lemma}

\begin{proof}
For every $u$ in the common domain,
\[
 \|D_1u\|\le\|D_0u\|+\kappa\|u\|,
\]
so
\[
 \|u\|+\|D_1u\|
 \le(1+\kappa)\|u\|+\|D_0u\|.
\]
Interchanging $D_0$ and $D_1$ gives
\[
 \|u\|+\|D_0u\|
 \le(1+\kappa)\|u\|+\|D_1u\|.
\]
Thus the two graph norms are equivalent.  Since the $D_0$ graph domain
embeds compactly into $\mathcal U$, so does the $D_1$ graph domain.
The min--max principle, with $E$ ranging over $j$-dimensional linear
subspaces of the common domain, gives
\[
 \sqrt{\mu_j^{(i)}}=
 \inf_{\substack{E\subset D(D_0)\\\dim E=j}}
 \ \sup_{\substack{u\in E\\\|u\|=1}}\|D_i u\|.
\]
For every unit vector $u$ in such a subspace $E$,
\[
 \|D_1u\|\le\|D_0u\|+\kappa.
\]
Hence
\[
 \sup_{\substack{u\in E\\\|u\|=1}}\|D_1u\|
 \le
 \sup_{\substack{u\in E\\\|u\|=1}}\|D_0u\|+\kappa.
\]
Taking the infimum over all $j$-dimensional $E$ gives
\[
 \sqrt{\mu_j^{(1)}}\le\sqrt{\mu_j^{(0)}}+\kappa.
\]
Interchanging $D_0$ and $D_1$ yields
\[
 \sqrt{\mu_j^{(0)}}\le\sqrt{\mu_j^{(1)}}+\kappa.
\]
Combining the two inequalities proves \eqref{eq:column-minmax}.
\end{proof}

For a real number $r$, we write $[r]_+=\max\{r,0\}$.
Whenever $-L_{\infty,H}$ has compact resolvent, we denote its ordered
eigenvalues, including multiplicity, by $\mu_j(H)$, $j\ge1$.
Applying the comparison to the second eigenvalue gives the following
bound.

\begin{theorem}
\label{thm:bounded-commutator-gap}
Under \eqref{eq:bounded-commutator-data}, the coherent generator has
compact resolvent.  For $0<\sigma_E\le\infty$,
\begin{equation}
\label{eq:bounded-commutator-gap}
 \operatorname{gap}(L_{\sigma_E,H})\ge
 \Lambda_W:=\left[
 \sqrt\lambda
 -2\sqrt{I_{\beta,f}}\,C_*(W)\right]_+^2.
\end{equation}
If $\Lambda_W>0$, then $\ker L_{\sigma_E,H}=\mathbb Cg_H$ and
$\rho_H$ is the unique invariant normal state.  For a density operator
$\rho=\rho_H^{1/4}X\rho_H^{1/4}$ with $X\in\mathscr T_2$,
\begin{equation}
\label{eq:bounded-dynamical}
 \|e^{t\mathcal L_{\sigma_E,H}}(\rho)-\rho_H\|_1
 \le e^{-\Lambda_Wt}\|X-g_H\|_2.
\end{equation}
\end{theorem}

\begin{proof}
We apply Lemma~\ref{lem:column-minmax} to the columns in
Lemma~\ref{lem:smooth-common-core}, with
$\kappa=2\sqrt{I_{\beta,f}}C_*(W)$.  Compactness of the free form embedding
was proved in \eqref{eq:reference-diagonalization}.  Set $\mathsf A_H:=-L_{\infty,H}$.  Then $\mathsf A_H$ has compact
resolvent, and Lemma~\ref{lem:column-minmax} gives, for every $j\ge1$,
\[
 \big|\sqrt{\mu_j(H)}-\sqrt{\mu_j(H_{0,n})}\big|
 \le2\sqrt{I_{\beta,f}}C_*(W).
\]
By \eqref{eq:smooth-ground-defect},
\[
 \mu_1(H)=0,
\]
while the free calculation gives
\[
 \mu_1(H_{0,n})=0,
 \qquad
 \mu_2(H_{0,n})=\lambda.
\]
Therefore
\[
 \sqrt{\mu_2(H)}
 \ge\sqrt\lambda-2\sqrt{I_{\beta,f}}C_*(W).
\]
Since $\mu_2(H)\ge0$, this is equivalent to
\[
 \mu_2(H)
 \ge[\sqrt\lambda-2\sqrt{I_{\beta,f}}C_*(W)]_+^2.
\]
When positive, this bound also proves that the kernel is one-dimensional;
otherwise the lower bound in \eqref{eq:bounded-commutator-gap} is zero.
This proves the estimate for the coherent sampler.  Only the commutators
need to be small; no smallness condition on $\|W\|$ has been used.
For real $s$, we define the unitary conjugation
$U_sX=e^{isH}Xe^{-isH}$ on $\mathscr T_2$.  For finite $\sigma_E>0$,
\cite[Gaussian averaging identity and gap monotonicity]{BeckerRouzeSalzmannToAppearcmp}
gives the Gaussian averaging identity
\begin{equation}
\label{eq:smooth-gaussian-average}
 \mathcal E_{\sigma_E,H}(X)
 =\int_{\mathbb R}p_{\sigma_E}(s)\mathcal E_H(U_{-s}X)\,ds.
\end{equation}
The identity holds first on the spectral core and then on the closed
form domain.
Let $X\perp g_H$.  Since $U_sg_H=g_H$ and $U_s$ is unitary,
\[
 \langle U_{-s}X,g_H\rangle_2
 =\langle X,U_sg_H\rangle_2=0,
 \qquad
 \|U_{-s}X\|_2=\|X\|_2.
\]
The coherent Poincar\'e estimate therefore gives
\[
 \mathcal E_H(U_{-s}X)\ge\Lambda_W\|X\|_2^2.
\]
Using $\int_{\mathbb R}p_{\sigma_E}(s)\,ds=1$ in
\eqref{eq:smooth-gaussian-average}, we obtain
\[
 \mathcal E_{\sigma_E,H}(X)
 \ge\int_{\mathbb R}p_{\sigma_E}(s)\Lambda_W\|X\|_2^2\,ds
 =\Lambda_W\|X\|_2^2.
\]
Thus finite Gaussian averaging preserves the same Poincar\'e lower bound.
Equation \eqref{eq:smooth-ground-defect} identifies $g_H$ as a kernel
vector; if $\Lambda_W>0$, the Poincar\'e estimate makes the kernel
one-dimensional.

Assume now that $\Lambda_W>0$.  We prove uniqueness among all normal
invariant states, not only those in the image of the Hilbert--Schmidt
embedding.
Let $\tau$ be such a state and let $E_k$ be finite spectral projections of
$\rho_H$ increasing to the identity.  Set
\[
 y_k=E_k\tau E_k,\qquad
 X_k=\rho_H^{-1/4}y_k\rho_H^{-1/4},
\]
where the inverse powers are restricted to $E_k\mathcal H_n$.
Then $y_k\to\tau$ in trace norm and $\operatorname{Tr}y_k\to1$.
For averaging time $T>0$, the Hilbert-space mean ergodic theorem gives,
as $T\to\infty$,
\[
 \frac1T\int_0^T e^{tL_{\sigma_E,H}}X_k\,dt
 \longrightarrow(\operatorname{Tr}y_k)g_H
 \quad\hbox{in }\mathscr T_2.
\]
The embedding $Z\mapsto\rho_H^{1/4}Z\rho_H^{1/4}$ is contractive from
$\mathscr T_2$ to $\mathscr T_1$.  We use \eqref{eq:sampler-intertwining}, the invariance of $\tau$, and
trace-norm contractivity to obtain
\[
 \|\tau-(\operatorname{Tr}y_k)\rho_H\|_1\le\|\tau-y_k\|_1.
\]
Letting $k\to\infty$ proves $\tau=\rho_H$.  Finally, the spectral theorem
on $g_H^\perp$ and the same contractive embedding prove
\eqref{eq:bounded-dynamical}.
\end{proof}

\subsection{Spectral gap for finite-rank perturbations of quadratic Hamiltonians}
\label{sec:GapForFiniteRankPerturb}

We return to $H=H_{n,M}=H_{0,n}+W_{n,M}$ with the global projection
$\Pi_M=P_M^{\otimes n}$.  The next statement holds for every bounded
self-adjoint $W=\Pi_MW\Pi_M$, regardless of its norm.

\begin{theorem}
\label{thm:finite-rank-kernel}
The coherent generator $L_{\infty,H}$ has compact resolvent.  If
$\widehat f(\nu)\ne0$ for every Bohr frequency of $H$, then
\[
 \ker L_{\infty,H}=\mathbb Cg_H,
 \qquad \operatorname{gap}(L_{\sigma_E,H})>0
 \quad(0<\sigma_E\le\infty).
\]
\end{theorem}

\begin{proof}
Since $W=\Pi_MW\Pi_M$ and $[H_{0,n},\Pi_M]=0$, we have
\[
 WA=W\Pi_MA,
 \qquad
 AW=A\Pi_MW.
\]
The range of $\Pi_M$ lies in a finite-energy oscillator subspace, so both
$\Pi_MA$ and $A\Pi_M$ are bounded.  Consequently
\[
 [W,A]=W\Pi_MA-A\Pi_MW
\]
and
\[
 \|[W,A]\|
 \le\|W\|\bigl(\|\Pi_MA\|+\|A\Pi_M\|\bigr)<\infty.
\]
Thus the hypotheses of Theorem~\ref{thm:bounded-commutator-gap} hold,
which gives compact resolvent.

To prove that the kernel is one-dimensional, we work on vectors with
finite oscillator spectral support.
For a multi-index $\mathbf m=(m_1,\ldots,m_{dn})\in\mathbb N_0^{dn}$,
let $e_{\mathbf m}=\bigotimes_{r=1}^{dn}|m_r\rangle$ be the multimode
Hermite vector.  We define $\mathcal F_{\rm fin}$ as the linear span of
these vectors.  We write $\mathbf e_r$ for the $r$th coordinate unit
vector in $\mathbb Z^{dn}$, and $\mathbf0$ for the zero multi-index.  Because $W$ is supported in a finite $H_{0,n}$-spectral subspace,
$\mathcal F_{\rm fin}$ is also the union of the finite $H$-spectral
subspaces.  Both $e^{\pm\beta H/4}$ and every bare ladder map this space
into itself.

Suppose $\mathcal E_H(X)=0$.  Then every component of the defect column
vanishes in $L^2(w_\beta(t)dt;\mathscr T_2)$:
\[
 D_H^A(t;X)=0
 \quad\text{for almost every }t,
 \qquad A\in\mathcal A.
\]
Fix $\phi,\psi\in\mathcal F_{\rm fin}$.  Since these vectors have finite
$H$-spectral support, there is a finite set $\Omega_{\phi,\psi,A}$ of
Bohr frequencies such that
\[
 \langle\phi,D_H^A(t;X)\psi\rangle
 =\sum_{\nu\in\Omega_{\phi,\psi,A}}
   e^{it\nu}\widehat f(\nu)
   \langle\phi,(A_\nu X-e^{\beta\nu/2}XA_\nu)\psi\rangle.
\]
Approximation of $X$ in the closed column norm, as in
Lemma~\ref{lem:smooth-common-core}, justifies this identity for the given
$X$.  The left-hand side vanishes in $L^2(w_\beta(t)dt)$, hence almost
everywhere because $w_\beta(t)>0$.  A finite linear combination of
distinct exponentials that vanishes on a set of positive measure has all
coefficients equal to zero.  Since $\widehat f(\nu)\ne0$, we obtain, for
each $\nu$,
\[
 \langle\phi,(A_\nu X-e^{\beta\nu/2}XA_\nu)\psi\rangle=0.
\]
Thus the weak identities
\begin{equation}
\label{eq:finite-rank-frequency-kernel}
 A_\nu X=e^{\beta\nu/2}XA_\nu
 \quad\hbox{on }\mathcal F_{\rm fin}\times\mathcal F_{\rm fin}.
\end{equation}
Define the sesquilinear form
\(B_X(\phi,\psi)= \langle e^{\beta H/4}\phi,Xe^{\beta H/4}\psi\rangle.\)
Equation \eqref{eq:finite-rank-frequency-kernel}, conjugation by
$e^{\beta H/4}$, and the finite sums on these test vectors give, for each
mode $r$,
\[
 B_X(a_r^\dagger\phi,\psi)=B_X(\phi,a_r\psi),
 \qquad B_X(a_r\phi,\psi)=B_X(\phi,a_r^\dagger\psi).
\]
For multi-indices $\mathbf m,\mathbf n\in\mathbb N_0^{dn}$, define
$b_{\mathbf m,\mathbf n}=B_X(e_{\mathbf m},e_{\mathbf n})$.
We set coefficients with a negative index equal to zero.  The ladder
identities give
\begin{align*}
 \sqrt{m_r+1}\,b_{\mathbf m+\mathbf e_r,\mathbf n}
 &=\sqrt{n_r}\,b_{\mathbf m,\mathbf n-\mathbf e_r},\\
 \sqrt{m_r}\,b_{\mathbf m-\mathbf e_r,\mathbf n}
 &=\sqrt{n_r+1}\,b_{\mathbf m,\mathbf n+\mathbf e_r}.
\end{align*}
Suppose first that $m_r>n_r$.  Repeated use of the first recurrence gives
\[
 b_{\mathbf m,\mathbf n}
 =C_{\mathbf m,\mathbf n,r}
   b_{\mathbf m-(n_r+1)\mathbf e_r,
     \mathbf n-(n_r+1)\mathbf e_r}=0,
\]
where $C_{\mathbf m,\mathbf n,r}$ is the nonzero product of the
square-root factors generated by the iteration.  The second multi-index on
the right has a negative $r$th component and is defined to have coefficient
zero.  If $m_r<n_r$, the
second recurrence gives in the same way
\[
 b_{\mathbf m,\mathbf n}=0.
\]
Thus $b_{\mathbf m,\mathbf n}=0$ whenever the two multi-indices differ in
at least one coordinate.  On the diagonal, setting
$\mathbf n=\mathbf m+\mathbf e_r$ in the first recurrence gives
\[
 \sqrt{m_r+1}\,b_{\mathbf m+\mathbf e_r,
                         \mathbf m+\mathbf e_r}
 =\sqrt{m_r+1}\,b_{\mathbf m,\mathbf m},
\]
so
\[
 b_{\mathbf m+\mathbf e_r,\mathbf m+\mathbf e_r}
 =b_{\mathbf m,\mathbf m}.
\]  Thus
$b_{\mathbf m,\mathbf n}=c\delta_{\mathbf m,\mathbf n}$, where
$c=b_{\mathbf0,\mathbf0}\in\mathbb C$ and
$\delta_{\mathbf m,\mathbf n}$ is one when the indices agree and zero
otherwise.  Since $e^{\beta H/4}$ maps $\mathcal F_{\rm fin}$ bijectively to itself,
$X=ce^{-\beta H/2}$.  Compact resolvent now makes zero an isolated simple
eigenvalue.  Equation \eqref{eq:smooth-gaussian-average} transfers its
positive gap to every finite $\sigma_E>0$.
\end{proof}

Theorem~\ref{thm:finite-rank-kernel} proves the fixed-Hamiltonian assertion of Theorem~\ref{thm:cutoff-gaps}.
This proof, together with the approximation of normal states in
Theorem~\ref{thm:bounded-commutator-gap}, also proves the uniqueness
assertion in Proposition~\ref{prop:sampler-generation}.

\begin{lemma}
\label{lem:path-gap-continuity}
Fix $n,M$, a bounded self-adjoint $V=\Pi_MV\Pi_M$, and the filter in
\eqref{eq:KMSFilterFunction}.  For $H(s)=H_{0,n}+sV$, $0\le s\le1$, the
coherent forms have a common domain, and the square root of each ordered eigenvalue of $-L_{\infty,H(s)}$ is
Lipschitz continuous in $s$.
In particular,
\begin{equation}
\label{eq:path-gap-positive}
 \lambda_*:=\min_{0\le s\le1}
       \operatorname{gap}(L_{\infty,H(s)})>0,
 \qquad
 \operatorname{gap}(L_{\sigma_E,H(s)})\ge\lambda_*
 \quad(0<\sigma_E\le\infty).
\end{equation}
\end{lemma}

\begin{proof}
Set $b=\|V\|$ and $C_A=\|[V,A]\|$.
Applying \eqref{eq:smooth-duhamel} with $sV$, and subtracting the analogous
formula with $rV$, for $r,s\in[0,1]$ and $u\in\mathbb R$, gives
\begin{equation}
\label{eq:path-duhamel-lipschitz}
 \|e^{iuH(s)}Ae^{-iuH(s)}-e^{iuH(r)}Ae^{-iuH(r)}\|
 \le |s-r|C_A(|u|+b|u|^2).
\end{equation}
For a bounded operator $T$, Duhamel's formula gives
\[
 e^{ivH(s)}Te^{-ivH(s)}-e^{ivH(r)}Te^{-ivH(r)}
 =i(s-r)\int_0^v
 e^{i(v-q)H(s)}[V,e^{iqH(r)}Te^{-iqH(r)}]e^{-i(v-q)H(s)}\,dq.
\]
Hence
\[
 \|e^{ivH(s)}Te^{-ivH(s)}-e^{ivH(r)}Te^{-ivH(r)}\|
 \le2|v||s-r|b\|T\|.
\]
Applying this estimate inside the Duhamel representation of the ladder
evolution yields the two contributions
\[
 |s-r|C_A|u|
 \quad\text{and}\quad
 |s-r|bC_A|u|^2,
\]
which gives \eqref{eq:path-duhamel-lipschitz}.

Since $f$ is Schwartz, the moments needed to integrate
$|u|+b|u|^2$ against $|f|$ are finite.  Thus
\eqref{eq:smooth-responses} gives
\[
 \|K_{sV}-K_{rV}\|\le C|s-r|
\]
for a constant $C$ depending only on the fixed data.  Applying
Lemma~\ref{lem:column-minmax} to $\mathfrak D_{H(s)}$ and
$\mathfrak D_{H(r)}$ gives, for every $j\ge1$,
\[
 \big|\sqrt{\mu_j(H(s))}-\sqrt{\mu_j(H(r))}\big|
 \le C|s-r|.
\]
Theorem~\ref{thm:finite-rank-kernel} makes $\mu_2(H(s))$ strictly
positive for every $s$.  Its continuity and compactness of $[0,1]$ give
the positive minimum in \eqref{eq:path-gap-positive}.
Equation \eqref{eq:smooth-gaussian-average} gives the same lower bound
for finite $\sigma_E$.
\end{proof}

\smallsection{Bounded Lipschitz pair potentials}
\label{sec:smooth-potential-gap}

For a multiplication operator, its commutator with a ladder operator is
a derivative.  The preceding gap estimate therefore applies directly to
bounded potentials with bounded first derivatives.  More precisely,
for any real bounded Lipschitz multiplication potential $W$ on
$\mathbb R^{dn}$, $\sum_A T_A^*T_A=\sum_A T_AT_A^*=|\nabla W|^2$.
Consequently
\begin{equation}\label{eq:intrinsic-force-gap}
 \operatorname{gap}(L_{\sigma_E,H_{0,n}+W})\ge
 \left[\sqrt{\lambda_{\rm OU}}
             -2\sqrt{I_{\beta,f}}\,\|\nabla W\|_\infty\right]_+^2.
\end{equation}
For $v\in W^{1,\infty}(\mathbb R^d;\mathbb R)$, we define
\[
 W_n^{(v)}=\sum_{i<j}\alpha_{n,i,j}v(x_i-x_j),\qquad
 H_n^{(v)}=H_{0,n}+W_n^{(v)},\qquad
 R_n=\left[\sum_i\left(\sum_{j\ne i}|\alpha_{n,i,j}|\right)^2\right]^{1/2}.
\]
The norm $\|\nabla v\|_\infty$ is the essential supremum of its Euclidean
length.

\begin{theorem}\label{thm:smooth-potential-gap}
For $0<\sigma_E\le\infty$,
\begin{equation}
\label{eq:smooth-potential-gap}
 \operatorname{gap}(L_{\sigma_E,H_n^{(v)}})\ge\lambda_{n,v}:=
 \left[\sqrt{\lambda_{\rm OU}}
       -2\sqrt{I_{\beta,f}}\,\|\nabla v\|_\infty R_n\right]_+^2.
\end{equation}
The total bound $B_n\|v\|_\infty$ need not be small.
If $\lambda_{n,v}>0$, the Gibbs state is the unique invariant normal
state, and the trace-norm convergence estimate
\eqref{eq:bounded-dynamical} holds with $\Lambda_W=\lambda_{n,v}$.
In particular, for any constant $J\ge0$, if
$|\alpha_{n,i,j}|\le J/n^{3/2}$, then
\begin{equation}
\label{eq:smooth-gap-n32}
 \operatorname{gap}(L_{\sigma_E,H_n^{(v)}})
 \ge\left[\sqrt{\lambda_{\rm OU}}
          -2J\sqrt{I_{\beta,f}}\,\|\nabla v\|_\infty\right]_+^2
\end{equation}
uniformly in $n$ and along $H_{0,n}+sW_n^{(v)}$, $0\le s\le1$.
\end{theorem}

\begin{proof}
Writing $W=W_n^{(v)}$, multiplication by $v$ gives
\[
 \|W\|\le\sum_{i<j}|\alpha_{n,i,j}|\,\|v\|_\infty
 =B_n\|v\|_\infty.
\]
For a fixed particle $i$ and coordinate $k$, only terms containing $i$
contribute to the derivative:
\[
 \partial_{x_{ik}}W
 =\sum_{j>i}\alpha_{n,i,j}\partial_k v(x_i-x_j)
  -\sum_{j<i}\alpha_{n,j,i}\partial_k v(x_j-x_i),
\]
where the symmetric convention for $\alpha_{n,i,j}$ is used.  Since
$a_{ik}=(x_{ik}+\partial_{x_{ik}})/\sqrt2$,
\[
 [W,a_{ik}]=-\frac1{\sqrt2}\partial_{x_{ik}}W,
 \qquad
 [W,a_{ik}^\dagger]=\frac1{\sqrt2}\partial_{x_{ik}}W.
\]
Therefore
\[
 \|[W,a_{ik}]\|,
 \ \|[W,a_{ik}^\dagger]\|
 \le\frac{\|\nabla v\|_\infty}{\sqrt2}
       \sum_{j\ne i}|\alpha_{n,i,j}|.
\]
Using the Euclidean gradient for each particle before taking the
supremum gives
\begin{equation}
\label{eq:smooth-column-bound}
 C_*(W)^2=\|\nabla W\|_\infty^2
 \le\|\nabla v\|_\infty^2
       \sum_i\left(\sum_{j\ne i}|\alpha_{n,i,j}|\right)^2
       =\|\nabla v\|_\infty^2R_n^2.
\end{equation}
Theorem~\ref{thm:bounded-commutator-gap} and \eqref{eq:smooth-column-bound}
prove \eqref{eq:smooth-potential-gap} and the convergence assertions.
Under $|\alpha_{n,i,j}|\le J/n^{3/2}$, for every $i$,
\[
 \sum_{j\ne i}|\alpha_{n,i,j}|
 \le\frac{J(n-1)}{n^{3/2}}.
\]
Hence
\[
 \begin{aligned}
 R_n^2
 &=\sum_{i=1}^n\left(\sum_{j\ne i}|\alpha_{n,i,j}|\right)^2\\
 &\le n\left(\frac{J(n-1)}{n^{3/2}}\right)^2
 =J^2\left(1-\frac1n\right)^2
 \le J^2.
 \end{aligned}
\]
Substitution into \eqref{eq:smooth-potential-gap} proves
\eqref{eq:smooth-gap-n32}.  Replacing $W$ by $sW$ multiplies every
commutator by $s\in[0,1]$, so the same lower bound holds along the whole
interpolation path.
\end{proof}

For the smaller scaling $|\alpha_{n,i,j}|\le J/n^2$, we have
$R_n\le J/\sqrt n$, so
\begin{equation}
\label{eq:smooth-gap-n2}
 \lambda_{n,v}\ge
 \left[\sqrt\lambda
       -2J\|\nabla v\|_\infty\sqrt{I_{\beta,f}/n}\right]_+^2.
\end{equation}
For each fixed $J$ and $v$, the right-hand side tends to $\lambda$ as
$n\to\infty$.
For a regularization length $\epsilon>0$, the three-dimensional potential
$v_\epsilon(x)=(|x|^2+\epsilon^2)^{-1/2}$,
\[
 \|v_\epsilon\|_\infty=\epsilon^{-1},\qquad
 \|\nabla v_\epsilon\|_\infty
 =\max_{r\ge0}\frac{r}{(r^2+\epsilon^2)^{3/2}}
 =\frac2{3\sqrt3\,\epsilon^2}.
\]
Thus \eqref{eq:smooth-gap-n2} gives
\[
 \lambda_{n,v_\epsilon}\ge
 \left[\sqrt\lambda-\frac{4J}{3\sqrt3\,\epsilon^2}
                                      \sqrt{I_{\beta,f}/n}\right]_+^2.
\]
This estimate holds for fixed positive regularization.  It does not
give a gap bound for the singular Coulomb potential when
$\epsilon\downarrow0$.

The Gaussian response
\begin{equation}
\label{eq:smooth-gaussian-filter}
 \widehat f_\ell(\nu)=
 \exp\left(-\frac{\nu^2}{4\ell^2}-\frac{\beta\nu}{4}\right),
 \qquad \ell>0,
\end{equation}
satisfies \eqref{eq:KMSFilterFunction} and is everywhere nonzero.
Fourier inversion gives
\[
 f_\ell(t)=\frac{\ell}{\sqrt\pi}
 e^{\ell^2\beta^2/16}e^{-\ell^2t^2}e^{i\ell^2\beta t/2},
 \qquad
 m_0(f_\ell)=e^{\ell^2\beta^2/16},\quad
 m_1(f_\ell)=\frac{e^{\ell^2\beta^2/16}}{\ell\sqrt\pi}.
\]
We define Catalan's constant by
$G_{\rm Cat}=\sum_{j=0}^\infty(-1)^j/(2j+1)^2$.
Direct integration gives
\begin{align}
 I_{\beta,f_\ell}
 &=e^{\ell^2\beta^2/8}
   \left(\frac{\beta^2}{32}
   +\frac{2\beta G_{\rm Cat}}{\pi^{5/2}\ell}
   +\frac{1}{2\pi\ell^2}\right),
 \label{eq:smooth-explicit-I}\\
 \lambda
 &=\frac{1-e^{-2\beta}}2
       \exp\left(\beta-\frac2{\ell^2}\right).
 \label{eq:smooth-explicit-ou}
\end{align}

\subsection{Uniform gaps for weakly interacting gases}
\label{sec:uniform_gap}

We keep the Hamiltonian $H_{n,M}=H_{0,n}+W_{n,M}$ with
$W_{n,M}=\Pi_MW_n\Pi_M$ and $\Pi_M=P_M^{\otimes n}$.
Recall that $\lambda_M$ is the largest one-particle energy retained by
$P_M$, and set
\begin{equation}
\label{eq:global-cutoff-data}
 j_M=\frac{\lambda_M-d}{2},\qquad
 \ell_M=\sqrt{j_M+1},\qquad
 R_M=\mathbf1_{(-\infty,\lambda_M+2]}(h),\qquad
 \widehat\Pi_{n,M}=R_M^{\otimes n}.
\end{equation}
We use the larger projection $\widehat\Pi_{n,M}$ only to bound the
change in the jumps.  The interaction remains
$\Pi_MW_n\Pi_M$ throughout.
With $\rho_1=\sigma_\beta(h)$ and $q=e^{-2\beta}$, define
\begin{equation}
\label{eq:global-thermal-mass}
 p_M=\operatorname{Tr}(\rho_1R_M)
 =(1-q)^d\sum_{k=0}^{j_M+1}\binom{k+d-1}{d-1}q^k<1.
\end{equation}
Thus $p_M$ depends only on $d,\beta,M$.  For the Gaussian filter
\eqref{eq:smooth-gaussian-filter}, equations
\eqref{eq:smooth-explicit-I}--\eqref{eq:smooth-explicit-ou} give
explicit values of the remaining filter constants.

\begin{lemma}
\label{lem:global-projection-resolvent}
For every $n\ge1$ and $X\in D(\mathfrak D_0)$,
\begin{equation}
\label{eq:global-projection-form-bound}
 (n+1)\big(\|\widehat\Pi_{n,M}X\|_2^2
              +\|X\widehat\Pi_{n,M}\|_2^2\big)
 \le\frac{2}{\lambda(1-p_M)}
          \big(\mathcal E_0(X)+\lambda\|X\|_2^2\big).
\end{equation}
\end{lemma}

\begin{proof}
On the one-particle Hilbert--Schmidt space, let
$\mathsf a=-L_h$, $g_1=\rho_1^{1/2}$, and
$\mathsf G=|g_1\rangle\langle g_1|$.
By \eqref{eq:smooth-free-poincare}, the one-particle Poincar\'e inequality says that $\mathsf a$ vanishes on
$\mathbb Cg_1$ and satisfies $\mathsf a\ge\lambda$ on $g_1^\perp$.
Therefore the spectral theorem gives
\[
 0\le e^{-t\mathsf a}
 \le \mathsf G+e^{-\lambda t}(I-\mathsf G).
\]
Let $S_\ell X=R_MX$ and $S_rX=XR_M$.  Both are orthogonal projections.
By \eqref{eq:global-thermal-mass},
\[
 \|S_\ell g_1\|_2^2
 =\operatorname{Tr}(\rho_1R_M)=p_M,
 \qquad
 \|S_rg_1\|_2^2=p_M.
\]
Compressing the semigroup inequality by $S_\ell$ gives
\[
 \begin{aligned}
 0\le S_\ell e^{-t\mathsf a}S_\ell
 &\le S_\ell\mathsf G S_\ell
  +e^{-\lambda t}S_\ell(I-\mathsf G)S_\ell.
 \end{aligned}
\]
Since
\[
 \|S_\ell\mathsf G S_\ell\|=\|S_\ell g_1\|_2^2=p_M,
\]
and $S_\ell(I-\mathsf G)S_\ell\le S_\ell$, we obtain
\[
 \|S_\ell e^{-t\mathsf a}S_\ell\|
 \le p_M+(1-p_M)e^{-\lambda t}.
\]
The same computation gives
\[
 \|S_re^{-t\mathsf a}S_r\|
 \le p_M+(1-p_M)e^{-\lambda t}.
\]
Under $\mathscr T_2(\mathcal H_n)=
\mathscr T_2(L^2(\mathbb R^d))^{\otimes n}$,
$e^{-t\mathsf A_{0,n}}=(e^{-t\mathsf a})^{\otimes n}$.
For $\mathsf S_\ell=S_\ell^{\otimes n}$, the Laplace integral for the
resolvent gives
\begin{align}
 \|\mathsf S_\ell(\mathsf A_{0,n}+\lambda)^{-1}
                     \mathsf S_\ell\|
 &\le\int_0^\infty e^{-\lambda t}
          \big(p_M+(1-p_M)e^{-\lambda t}\big)^n\,dt\nonumber\\
 &=\frac1\lambda\int_0^1\big(p_M+(1-p_M)y\big)^n\,dy
 =\frac{1-p_M^{n+1}}{\lambda(n+1)(1-p_M)}.
 \label{eq:global-resolvent-integral}
\end{align}
Using
$\|\mathsf S_\ell(\mathsf A_{0,n}+\lambda)^{-1/2}\|^2
=\|\mathsf S_\ell(\mathsf A_{0,n}+\lambda)^{-1}\mathsf S_\ell\|$
in \eqref{eq:global-resolvent-integral} proves
\[
 \|\widehat\Pi_{n,M}X\|_2^2
 \le\frac{\mathcal E_0(X)+\lambda\|X\|_2^2}
           {\lambda(n+1)(1-p_M)}.
\]
The same calculation with $\mathsf S_r=S_r^{\otimes n}$ bounds
$\|X\widehat\Pi_{n,M}\|_2^2$.  Adding the two bounds proves
\eqref{eq:global-projection-form-bound}.
\end{proof}

\begin{lemma}
\label{lem:global-cutoff-column-small}
Let $W=W^*=\Pi_MW\Pi_M$, $H=H_{0,n}+W$, and $b=\|W\|$.
Define
\begin{equation}
\label{eq:global-column-constant}
 C_M=\left(
 \frac{32d\ell_M^2 I_{\beta,f}}{\lambda(1-p_M)}
 \right)^{1/2}.
\end{equation}
The domains of $\mathfrak D_H$ and $\mathfrak D_0$ coincide.  For every
$X\in D(\mathfrak D_0)$,
\begin{equation}
\label{eq:global-column-bound}
 \|(\mathfrak D_H-\mathfrak D_0)X\|_{\mathfrak K_n}
 \le C_Mb\big(\mathcal E_0(X)+\lambda\|X\|_2^2\big)^{1/2}.
\end{equation}
The constant in \eqref{eq:global-column-bound} is independent of $n$.
\end{lemma}

\begin{proof}
Each ladder changes the excitation number of one particle by one.
Equation \eqref{eq:global-cutoff-data} therefore gives
\begin{gather*}
 \|A\Pi_M\|,\ \|\Pi_MA\|\le\ell_M,
 \qquad \|[W,A]\|\le2\ell_Mb,\\
 [W,A]=W\Pi_MA-A\Pi_MW
       =\widehat\Pi_{n,M}[W,A]\widehat\Pi_{n,M}.
\end{gather*}
Both $H_{0,n}$ and $W$ preserve the range of $\widehat\Pi_{n,M}$ and
its orthogonal complement, so $[H,\widehat\Pi_{n,M}]=0$.
The Duhamel formula \eqref{eq:smooth-duhamel} gives
\[
 T_A(u):=e^{iuH}Ae^{-iuH}-e^{iuH_{0,n}}Ae^{-iuH_{0,n}}
 =\widehat\Pi_{n,M}T_A(u)\widehat\Pi_{n,M},
 \qquad\|T_A(u)\|\le2\ell_Mb|u|.
\]
After integrating against $f$ and $\overline f$ in
\eqref{eq:smooth-responses}, both response differences vanish outside the range of
$\widehat\Pi_{n,M}$ and have norm at most $2\ell_Mb h_f(t)$.
Consequently, for $K_W=\mathfrak D_H-\mathfrak D_0$,
\begin{align}
 \|K_WX\|_{\mathfrak K_n}^2
 &\le 2\sum_{A\in\mathcal A}\int_{\mathbb R}w_\beta(t)
        (2\ell_Mb h_f(t))^2
        \big(\|\widehat\Pi_{n,M}X\|_2^2
                  +\|X\widehat\Pi_{n,M}\|_2^2\big)\,dt\nonumber\\
 &=16dn\ell_M^2b^2I_{\beta,f}
       \big(\|\widehat\Pi_{n,M}X\|_2^2
                  +\|X\widehat\Pi_{n,M}\|_2^2\big).
 \label{eq:global-supported-column}
\end{align}
There are $2dn$ bare jumps.  Substituting
\eqref{eq:global-projection-form-bound} into
\eqref{eq:global-supported-column} gives
\[
 \begin{aligned}
 \|K_WX\|_{\mathfrak K_n}^2
 &\le16dn\ell_M^2b^2I_{\beta,f}
 \frac{2}{\lambda(n+1)(1-p_M)}
 \bigl(\mathcal E_0(X)+\lambda\|X\|_2^2\bigr)\\
 &=\frac{32d\ell_M^2I_{\beta,f}}{\lambda(1-p_M)}
   \frac{n}{n+1}\,b^2
 \bigl(\mathcal E_0(X)+\lambda\|X\|_2^2\bigr)\\
 &\le C_M^2b^2
 \bigl(\mathcal E_0(X)+\lambda\|X\|_2^2\bigr).
 \end{aligned}
\]
Taking square roots gives \eqref{eq:global-column-bound}.  The first display verifies the
assumptions of Lemma~\ref{lem:smooth-common-core}, which gives equality
of the closed domains.
\end{proof}

\begin{theorem}
\label{thm:uniform-gap}
Fix $d\in\{2,3\}$, $\beta>0$, $M\in\mathbb N$, and the filter
$f$ in \eqref{eq:KMSFilterFunction}.  For $C_M$ in
\eqref{eq:global-column-constant} and every bounded self-adjoint
$W=\Pi_MW\Pi_M$,
\begin{equation}
\label{eq:uniform-quantitative-gap}
 \operatorname{gap}(L_{\sigma_E,H_{0,n}+sW})
 \ge\lambda_{\rm OU}[1-\sqrt2\,sC_M\|W\|]_+^2
 \qquad(0\le s\le1,\ 0<\sigma_E\le\infty).
\end{equation}
In particular, with
\begin{equation}
\label{eq:uniform-threshold}
 b_M^*=\frac1{2\sqrt2\,C_M},
\end{equation}
$\|W\|\le b_M^*$ implies
\begin{equation}
\label{eq:uniform-global-gap}
 \operatorname{gap}(L_{\sigma_E,H_{0,n}+sW})
 \ge\lambda_{\rm OU}/4
 \qquad(0\le s\le1,\ 0<\sigma_E\le\infty).
\end{equation}
For the globally truncated Coulomb interaction, define
\begin{equation}
\label{eq:uniform-coulomb-threshold}
 v_M=\|(P_M\otimes P_M)W_{12}(P_M\otimes P_M)\|,\qquad
 \varepsilon_M=\frac{2b_M^*}{1+v_M}.
\end{equation}
Then $|\alpha_{n,i,j}|\le\varepsilon_M/n^2$ implies
\eqref{eq:uniform-global-gap} for $W=W_{n,M}$, uniformly in $n$.
\end{theorem}

\begin{proof}
We first take $s=1$ and $\sigma_E=\infty$, and set
$\delta=C_M\|W\|$.  Compact resolvent follows from
Theorem~\ref{thm:bounded-commutator-gap}.  If $X\perp g_0$, then
$\mathcal E_0(X)\ge\lambda\|X\|_2^2$.  For $\delta<1/\sqrt2$,
\eqref{eq:global-column-bound} therefore gives
\[
 \begin{aligned}
 \mathcal E_H(X)^{1/2}
 &\ge\mathcal E_0(X)^{1/2}
       -\delta\bigl(\mathcal E_0(X)+\lambda\|X\|_2^2\bigr)^{1/2}\\
 &\ge(1-\sqrt2\,\delta)\mathcal E_0(X)^{1/2}
 \ge\sqrt\lambda(1-\sqrt2\,\delta)\|X\|_2.
 \end{aligned}
\]
Writing $\mu_2(H)$ for the second ordered eigenvalue of $-L_{\infty,H}$,
the max--min principle gives
\[
 \mu_2(H)\ge
 \inf_{\substack{X\in D(\mathfrak D_0),\ X\perp g_0\\\|X\|_2=1}}
           \mathcal E_H(X)
 \ge\lambda(1-\sqrt2\,\delta)^2.
\]
The vector $g_H$ lies in the kernel by \eqref{eq:smooth-ground-defect}.
Thus the positive second-eigenvalue bound also proves simplicity of zero.
For $\delta\ge1/\sqrt2$, the stated lower bound is zero.
Replacing $W$ by $sW$ and using the Gaussian average
\eqref{eq:smooth-gaussian-average} proves
\eqref{eq:uniform-quantitative-gap}.  The threshold
\eqref{eq:uniform-threshold} then gives \eqref{eq:uniform-global-gap}.

For $1\le r\le n$, we write $P_M^{(r)}$ for $P_M$ acting on particle
$r$ and as the identity on the others.  For a pair $i<j$, we keep all
these spectator projections.  This gives
\[
 \Pi_MW_{ij}\Pi_M
 =\left((P_M^{(i)}P_M^{(j)})W_{ij}
                    (P_M^{(i)}P_M^{(j)})\right)
       \prod_{r\ne i,j}P_M^{(r)},
 \qquad
 \|\Pi_MW_{ij}\Pi_M\|=v_M.
\]
Equation \eqref{eq:pair-matrix-bound} gives
\[
 v_M\le C_d\sqrt{\lambda_M}<\infty.
\]
Under $|\alpha_{n,i,j}|\le\varepsilon_M/n^2$,
\[
 B_n
 =\sum_{i<j}|\alpha_{n,i,j}|
 \le\binom n2\frac{\varepsilon_M}{n^2}
 \le\frac{\varepsilon_M}{2}.
\]
Thus
\[
 \|W_{n,M}\|\le v_MB_n
 \le v_M\binom n2\frac{\varepsilon_M}{n^2}
 \le\frac{v_M\varepsilon_M}{2}\le b_M^*.
\]
The same estimate holds for $sW_{n,M}$ for every $s\in[0,1]$.
\end{proof}

Together with Lemma~\ref{lem:path-gap-continuity}, this completes the
proof of Theorem~\ref{thm:cutoff-gaps}.  Here $M,\beta,f$ are fixed
before $n$ varies.  In particular, the factor $(1-p_M)^{-1/2}$
in \eqref{eq:global-column-constant} grows without bound as $M\to\infty$.
The theorem also applies whenever $v_MB_n\le b_M^*$, without assuming a particular
scaling for each pair coupling.
For $|\alpha_{n,i,j}|\le J/n^2$,
Theorem~\ref{thm:coulomb-approximation} gives
\[
 \left|\frac{F_\beta(H_n)}n-\frac{F_\beta(H_{n,M})}n\right|
 \le C_F\left(\frac J{\sqrt n}+\frac{J^3}{n^{5/2}}\right)M^{-1/(2d)}.
\]
A fixed cutoff controls this error uniformly in $n$; the corresponding
coupling threshold then controls the entire interpolation path.

\smallsection{Pairwise compression as an alternative regularization}
Let $P_i$ be $P_M$ on particle $i$ and the identity on the others, and set
\[
 V_{n,M}^{\rm pair}=\sum_{i<j}\alpha_{n,i,j}P_iP_jW_{ij}P_iP_j,
 \qquad H_{n,M}^{\rm pair}=H_{0,n}+V_{n,M}^{\rm pair}.
\]
This bounded interaction retains identity factors on spectators and is
usually not finite rank on $\mathcal H_n$.
\begin{corollary}\label{cor:pairwise-compression}
The free-energy, trace-distance, and symmetric relative-entropy bounds
\eqref{eq:main-free-energy-error}--\eqref{eq:main-entropy-error} hold for
$H_{n,M}^{\rm pair}$ with $\Xi_n$ replaced by
\[
 \Xi_n^{\rm pair}=A_n[n+(A_n+1)B_n].
\]
For repulsive three-dimensional interactions one may instead use
$A_n(n+B_n)$.  If $A_n\le A_*$ and $B_n\le B_*n$, then
\[
 \sup_n\left|\frac{F_\beta(H_n)-F_\beta(H_{n,M}^{\rm pair})}{n}\right|
 \le C_{\beta,d,A_*,B_*}M^{-1/(2d)}.
\]
For $R_n$ defined in Section~\ref{sec:smooth-potential-gap},
$\ell_M$ from \eqref{eq:global-cutoff-data}, and
$v_M=\|(P_M\otimes P_M)W_{12}(P_M\otimes P_M)\|$, one has
\begin{equation}\label{eq:pairwise-gap}
 \operatorname{gap}(L_{\sigma_E,H_{0,n}+sV_{n,M}^{\rm pair}})
 \ge\left[\sqrt{\lambda_{\rm OU}}
      -4s\sqrt{2dI_{\beta,f}}\,\ell_Mv_M R_n\right]_+^2
 \quad(0\le s\le1).
\end{equation}
At fixed $n,M$, the coherent generator has compact resolvent and simple
kernel; its gap has a positive minimum along this path, for arbitrary
couplings.  The same path minimum bounds every finite-$\sigma_E$ gap.
\end{corollary}
\begin{proof}
Apply Lemma~\ref{lem:global-coulomb-error} to one pair, with reference
energy $h_i+h_j$, and sum with weights $|\alpha_{n,i,j}|$.  Since
$\sum_{i<j}|\alpha_{n,i,j}|(h_i+h_j)\le A_nH_{0,n}$, this gives
\[
 |(W_n-V_{n,M}^{\rm pair})[\psi]|
 \le6A_n\lambda_{M+1}^{-1/2}\langle\psi,H_{0,n}\psi\rangle.
\]
The thermal proof applies termwise: each $P_iP_j$ commutes with
$H_{0,n}$ and the free Gibbs state, and compression preserves the pair
form bounds and repulsive positivity.  Gibbs stability proves the
three approximation estimates and their density consequence.
Only pairs containing particle $i$ contribute to its ladder commutator;
$\|P_i a_{ik}^{\#}\|,\|a_{ik}^{\#}P_i\|\le\ell_M$ for either ladder.
Hence
\[
 \|[V_{n,M}^{\rm pair},a_{ik}^{\#}]\|
 \le2\ell_Mv_M\sum_{j\ne i}|\alpha_{n,i,j}|,
 \qquad C_*(V_{n,M}^{\rm pair})\le2\sqrt{2d}\,\ell_Mv_M R_n.
\]
Theorem~\ref{thm:bounded-commutator-gap} proves
\eqref{eq:pairwise-gap} and compact resolvent.
For kernel simplicity, every $P_i$ reduces $H_{n,M}^{\rm pair}$.
In each sector specifying which particles lie in $P_i$, the interacting
part is finite dimensional and the remaining particles are free
oscillators.  Fixing their Hermite indices gives finite invariant blocks.
Thus the finite Fock span equals the union of finite Hamiltonian spectral
subspaces and is preserved by all ladders and $e^{\pm\beta H/4}$.
The frequency and ladder-recurrence proof of
Theorem~\ref{thm:finite-rank-kernel} applies verbatim.  The boundedness
and bounded commutators also give the path continuity proof of
Lemma~\ref{lem:path-gap-continuity}.
\end{proof}
Here $A_n\le B_n$, so the pairwise coefficient is no larger than the
global one.  Moreover $\ell_Mv_M=\mathcal O(M^{1/d})$; the explicit
smallness threshold in \eqref{eq:pairwise-gap} therefore deteriorates
only polynomially with $M$.  In particular, $|\alpha_{n,i,j}|\le J/n^{3/2}$
gives $R_n\le J$ and a uniform gap for fixed $M$ and sufficiently small
$J$.  The usual mean-field scale $J/n$ gives the free-energy density
estimate, but no uniform mixing gap is asserted at that scale.
These are static and spectral alternatives.  The circuits below retain
the global compression; their tail and occupation-cutoff arguments are
not being asserted for the pairwise regularization.

\section{Algorithmic implementations}
\label{sec:algorithm-input}
We now combine the approximation and gap estimates.  There are three
steps in state preparation: run the sampler long enough to mix, restrict
its dynamics to a finite occupation register, and approximate that finite
evolution by a circuit.  The interaction cutoff and the occupation
cutoff play different roles and are kept separate.

\smallsection{Input and resource model}
We fix $d\in\{2,3\}$ and positive $\beta,\ell,\sigma_E$ throughout
this section.  For every precision integer $p\ge1$, we assume these
parameters can be computed to $p$ bits in time polynomial in $p$.  The interaction rank $M\ge1$ is distinct from the
integer occupation cutoff $K\ge j_M+2$, where $j_M$ is defined in
\eqref{eq:global-cutoff-data}.  For the cutoff definitions, we allow any
bounded $V=V^*=\Pi_MV\Pi_M$ and set $H=H_{0,n}+V$; the algorithms use
$V=W_{n,M}$.  We define the one-mode and full occupation projections by
\(q_K=\sum_{m=0}^K|m\rangle\langle m|,\qquad \mathcal Q_K=q_K^{\otimes dn}.\)
For the mode $r=d(i-1)+k$, let $q_K^{(r)}$ act on that mode and as the
identity on all others.  We set
\begin{equation}
\label{eq:circuit-mode-cutoff}
 \begin{aligned}
 a_r^{\le K}&=q_K^{(r)}a_rq_K^{(r)},\qquad
 (a_r^\dagger)^{\le K}=(a_r^{\le K})^\dagger,\\
 H_{\le K}&=\left.\mathcal Q_KH\mathcal Q_K
                                \right|_{\mathcal Q_K\mathcal H_n}.
 \end{aligned}
\end{equation}
For a bare jump $A=a_r$ or $a_r^\dagger$, $A^{\le K}$ denotes the
corresponding truncated ladder.  Each has norm at most $\sqrt{K+1}$.
The operator $H_{\le K}$ acts on the finite register
$\mathcal Q_K\mathcal H_n$, not on its infinite-dimensional complement.
Since $\Pi_M\le\mathcal Q_K$, the interaction is unchanged.
The Hamiltonian, the truncated ladders, and their filtered jumps and
loss operators commute with $\mathcal Q_K$.  Their dynamics therefore
preserves this register.
We use the Gaussian response
\begin{equation*}
\widehat f_\ell(\nu)=
 e^{-\nu^2/(4\ell^2)-\beta\nu/4}.
\end{equation*}
It satisfies \eqref{eq:KMSFilterFunction}, is everywhere positive, and
belongs to the Schwartz class.  To check the filter regularity required by \cite[filter Gevrey condition]{BeckerRouzeSalzmannToAppearcmp}, put
$u(z)=\widehat f_\ell(z/\beta)$ for $z\in\mathbb C$, using the entire
Gaussian formula.  For each strip width $r>0$, we define
$C_r=\sup_{|\operatorname{Im}z|\le r}|u(z)|<\infty$, where
$\operatorname{Im}z$ is the imaginary part of $z$.
Cauchy's formula gives, for derivative order $k\in\mathbb N_0$,
\(\|u^{(k)}\|_{L^\infty(\mathbb R)}\le C_r r^{-k}k!, \qquad \|u'\|_{L^1(\mathbb R)}<\infty.\)
Thus the scaled response has Gevrey order one, with constants depending
only on $\beta,\ell$.

The real couplings and the Hermitian pair matrix
\begin{equation}
\label{eq:algorithm-pair-data}
 V_M^{(2)}=(P_M\otimes P_M)W_{12}(P_M\otimes P_M)
 \quad\hbox{on }\operatorname{Ran}(P_M\otimes P_M)
\end{equation}
are supplied to $p$-bit entrywise accuracy with classical cost polynomial
in $n,M,p$.  Its dimension is $M^2$.  An entrywise error $2^{-p}$
produces an operator-norm error at most $M^2 2^{-p}$; symmetrizing the
supplied matrix preserves this bound.  Coupling precision and subsequent
matrix operations are included in this input model.  We choose all circuit parameters using certified rational upper bounds.  In particular,
we compute
\begin{equation}
\label{eq:certified-input-bounds}
 B_n\le\overline B_n\le B_n+1,\qquad
 v_M\le\overline v_M\le v_M+1,\qquad
 \overline b=\overline B_n\overline v_M\ge\|W_{n,M}\|.
\end{equation}
For an accuracy $0<\varepsilon<1$, we also compute a rational enclosure
$\widehat\Xi_n\le\overline\Xi_n\le\widehat\Xi_n+\varepsilon$,
using certified upper enclosures for $A_n,B_n$ in the explicit formula
for $\widehat\Xi_n$.  Bisection for the square root and sufficiently
accurate coupling entries give this enclosure with polynomial cost;
indeed $|\sqrt x-\sqrt y|\le\sqrt{|x-y|}$ and the remaining factors
are polynomially bounded.  For the state cutoff use $\varepsilon^2$
in place of $\varepsilon$.  In the repulsive case, nonnegativity of the
couplings is part of the input promise.  These enclosures cover zero
couplings.  The other bounds follow from entrywise error certificates
and certified Hermitian matrix-norm computation.

For the scalar filters, we use the access model in
\cite[Circuit simulation of the quadrature generator]{BeckerRouzeSalzmannToAppearcmp}.  More explicitly,
let $\chi:\mathbb R\to\mathbb R$ be a fixed even smooth function,
equal to one on $[-1/2,1/2]$ and supported in $[-1,1]$, of Gevrey
order two.  This means that some constants $C_\chi,R_\chi>0$ satisfy
$\|\chi^{(j)}\|_\infty\le C_\chi R_\chi^j(j!)^2$ for all
$j\in\mathbb N_0$.  For a frequency cutoff $S>0$, we put
$\chi_S(\nu)=\chi(\nu/S)$.
We assume the coherent state-preparation oracles, including their
phase information, required by the cited circuit-simulation corollary
for the inverse Fourier transforms of
\begin{equation}
\label{eq:algorithm-scalar-data}
 \widehat f_\ell(\nu)\chi_S(\nu),\qquad
 -\frac i2\tanh(-\beta\nu/4)\chi_S(\nu).
\end{equation}
We also assume coherent reversible access to the Gauss--Hermite nodes
and weights to the required precision.  Evaluation of scalar coefficients
alone is not substituted for these state-preparation oracles.
The oracles used on the time grids of the cited corollary have gate cost
polynomial in $S$, the grid range, and the number $p$ of requested
precision bits.  We count their internal workspace separately.  For the finite Hamiltonian $H_{\le K}$, we choose
$S\ge\max\{1,4\|H_{\le K}\|\}$, so $\chi_S$ equals one at all its
Bohr frequencies.  The Gaussian derivative bounds and the fixed smooth cutoff satisfy
the window and filter Gevrey conditions of that reference.

A positive rational certified lower bound on the relevant spectral gap is
also supplied.  Theorems~\ref{thm:cutoff-gaps} and
\ref{thm:uniform-gap} establish respectively positivity and an explicit
bound in the small-coupling regime.  We count the qubits holding the particle occupations separately from
the simulation ancillas.  For simulated time $t\ge0$ and circuit
accuracy $0<\eta<1$, we denote the number of these ancillas by
$Q_{\rm anc}$ and have the bound
\begin{equation}
\label{eq:circuit-ancillas}
 Q_{\rm anc}=\widetilde{\mathcal O}\left(
       \log(2+nMK)+\log^2(2+nt/\eta)\right)
\end{equation}
in \cite[Circuit simulation of the quadrature generator]{BeckerRouzeSalzmannToAppearcmp}, in addition to
oracle workspace.  The $\mathcal O(n)$ reusable membership controls below can be
included in the $\mathcal O(n\log(K+1))$ system-scale overhead.  The output is stored classically, and repeated samples reuse the same
quantum registers.  The notation $\lceil x\rceil$ means the least integer
not smaller than a real number $x$, and $\log_2$ is the base-two logarithm.

\subsection{Preparing Gibbs states of trapped Coulomb gases}
\label{sec:gibbs-preparation}
We use the occupation projection and truncated ladders from
\eqref{eq:circuit-mode-cutoff}.  The Hamiltonian remains
$H=H_{0,n}+V$, with $V=V^*=\Pi_MV\Pi_M$.

We first compare the initial vacuum state with the Gibbs state.  We
then bound the transitions removed by the occupation cutoff.
For every bounded self-adjoint $V$, eigenvalue comparison with
$H_{0,n}\pm\|V\|$ gives
\begin{equation}
\label{eq:bounded-partition-comparison}
 e^{-\beta\|V\|}(2\sinh\beta)^{-dn}
 \le\mathcal Z_\beta(H_{0,n}+V)
 \le e^{\beta\|V\|}(2\sinh\beta)^{-dn}.
\end{equation}
This gives the constant needed for the mixing estimate
\eqref{eq:warm-start-mixing}.

\begin{lemma}
\label{lem:finite-rank-warm-start}
Let $d,n\ge1$ be integers and $H=H_{0,n}+V$ with
$V=V^*=\Pi_MV\Pi_M$.  Every normalized vector $\Psi\in\operatorname{Ran}\Pi_M$ satisfies
\[
 |\Psi\rangle\langle\Psi|
 \le e^{\beta(2\|V\|+n\lambda_M)}(2\sinh\beta)^{-dn}\sigma_\beta(H).
\]
In particular, this applies to the vacuum $\Psi=|0\rangle^{\otimes dn}$,
where $a|0\rangle=0$ for each one-dimensional oscillator.
\end{lemma}

\begin{proof}
For a rank-one state, the inequality
\[
 |\Psi\rangle\langle\Psi|\le\mathfrak c\,\sigma_\beta(H)
\]
is equivalent to
\[
 \mathfrak c\ge
 \langle\Psi,\sigma_\beta(H)^{-1}\Psi\rangle.
\]
Hence the smallest domination constant is
\[
 \mathfrak c
 =\langle\Psi,\sigma_\beta(H)^{-1}\Psi\rangle.
\]
Since
\[
 \sigma_\beta(H)^{-1}=\mathcal Z_\beta(H)e^{\beta H}
\]
on the range under consideration,
\[
 \mathfrak c
 =\mathcal Z_\beta(H)\langle\Psi,e^{\beta H}\Psi\rangle.
\]
The projection $\Pi_M$ commutes with $H_{0,n}$ and $V$, hence with $H$.
On $\operatorname{Ran}\Pi_M$,
\[
 H_{0,n}\le n\lambda_MI,
 \qquad
 -\|V\|I\le V\le\|V\|I,
\]
so
\[
 H\le(n\lambda_M+\|V\|)I.
\]
Functional calculus gives
\[
 \langle\Psi,e^{\beta H}\Psi\rangle
 \le e^{\beta(n\lambda_M+\|V\|)}.
\]
Together with the upper bound in
\eqref{eq:bounded-partition-comparison},
\[
 \mathcal Z_\beta(H)
 \le e^{\beta\|V\|}(2\sinh\beta)^{-dn},
\]
we obtain
\[
 \mathfrak c
 \le e^{\beta(2\|V\|+n\lambda_M)}(2\sinh\beta)^{-dn}.
\]
This proves the lemma.
\end{proof}

\begin{lemma}
\label{lem:direct-dynamical-cutoff}
Let $V=V^*=\Pi_MV\Pi_M$, $H=H_{0,n}+V$, $b=\|V\|$, and
$K\ge j_M+2$.  With $0<\sigma_E<\infty$ fixed, let
$\mathcal L_K$ be the bounded generator obtained from
\eqref{eq:loss-integral}--\eqref{eq:gain-integral} by replacing each
$A$ by $A^{\le K}$ and retaining $H$.  Its restriction to
$\mathcal Q_K\mathcal H_n$ is the finite-dimensional sampler for
$H_{\le K}$ and the truncated ladders.
For every density operator
$\rho_{\rm ini}=\mathcal Q_K\rho_{\rm ini}\mathcal Q_K
\le\mathfrak c\sigma_\beta(H)$, $\mathfrak c\ge1$, and $t\ge0$,
\begin{equation}
\label{eq:direct-dynamical-cutoff}
 \|(e^{t\mathcal L_{\sigma_E,H}}-e^{t\mathcal L_K})
                       (\rho_{\rm ini})\|_1
 \le C_{\beta,f,d}\,tn\mathfrak c e^{\beta b}(K+1)e^{-\beta K}.
\end{equation}
The constant is independent of $n,M,K,b,t$ and $0<\sigma_E<\infty$.
One explicit choice is
\begin{equation}
\label{eq:direct-cutoff-constant}
 \begin{aligned}
 q&=e^{-2\beta},\quad c_1=(1-q)^{-1},\quad
 c_2=(1+q)(1-q)^{-2},\\
 F&=m_0(f),\qquad v_f=\max_{\theta\in\{-2,2\}}|\widehat f(\theta)|,\\
 C_{\beta,f,d}&=2d\bigl(2Fv_f\sqrt{c_1}+v_f^2c_1+v_f^2\sqrt{c_2}\bigr).
 \end{aligned}
\end{equation}
\end{lemma}

\begin{proof}
We prove the operator identities on the oscillator core and then pass
to their closed extensions.
For $A=a_r$ or $a_r^\dagger$, put $D_A=A-A^{\le K}$ and
$\theta_A=-2$ or $2$, respectively.  Both endpoints of every omitted
transition have mode occupation at least $K$, so
$VD_A=D_AV=0$.  We define the filtered jumps
$J_A=L_H^A$ and $J_{A,K}=L_H^{A^{\le K}}$.  Thus
\begin{equation}
\label{eq:exact-tail-evolution}
 e^{iuH}D_Ae^{-iuH}=e^{i\theta_Au}D_A,\qquad
 J_A-J_{A,K}=\widehat f(\theta_A)D_A,
\end{equation}
We define the bounded remainder
$\mathcal R_A=J_A-\widehat f(\theta_A)A$.
The support identity in Lemma~\ref{lem:global-cutoff-column-small} gives
$\mathcal R_A=\widehat\Pi_{n,M}\mathcal R_A\widehat\Pi_{n,M}$.
Since $K\ge j_M+2$, $\mathcal R_A^\dagger D_A=0$.
The explicit weighted shifts in \eqref{eq:circuit-mode-cutoff} also give
$(A^{\le K})^\dagger D_A=0$.  Consequently
\begin{equation}
\label{eq:exact-tail-loss}
 J_{A,K}^\dagger D_A=0,\qquad
 J_A^\dagger J_A-J_{A,K}^\dagger J_{A,K}
       =|\widehat f(\theta_A)|^2D_A^\dagger D_A.
\end{equation}
The operator $D_A^\dagger D_A$ is a function of the occupation number of
one mode.  It commutes with $H_{0,n}$, and it vanishes on the support of
$V$; hence it also commutes with $H=H_{0,n}+V$.
Therefore the conjugation in the loss integral leaves it unchanged.
Using
\[
 \int_{\mathbb R}g_{\sigma_E}(t)\,dt=\frac12,
\]
the loss difference contributed by this jump is
\[
 -\int_{\mathbb R}g_{\sigma_E}(t)e^{itH}
   |\widehat f(\theta_A)|^2D_A^\dagger D_Ae^{-itH}\,dt
 =-\frac{|\widehat f(\theta_A)|^2}{2}D_A^\dagger D_A.
\]
The bounded jumps $J_{A,K}$ and the integrable loss kernel define a
bounded generator; its real loss part equals minus one half of the
Gaussian-averaged jump square.  The bounded Lindblad construction
therefore gives a completely positive trace-preserving semigroup.
Since these operators commute with $\mathcal Q_K$, the finite
occupation register is invariant under this dynamics.

We next bound the omitted transitions in the Gibbs state
$\rho_H=\sigma_\beta(H)$.
Writing $N_r=a_r^\dagger a_r$, the omitted shifts satisfy
\[
 D_{a_r}^\dagger D_{a_r}=N_r\mathbf1_{[K+1,\infty)}(N_r),\qquad
 D_{a_r^\dagger}^\dagger D_{a_r^\dagger}
                =(N_r+1)\mathbf1_{[K,\infty)}(N_r).
\]
These diagonal operators and their squares vanish on $\Pi_M$.
The operators $(D_A^\dagger D_A)^j$ vanish on $\Pi_M$, so only the
orthogonal complement contributes.  On $\Pi_M^\perp$ the perturbation
$V$ vanishes and $H=H_{0,n}$.  Hence
\[
 \begin{aligned}
 \operatorname{Tr}\bigl(\rho_H(D_A^\dagger D_A)^j\bigr)
 &=\mathcal Z_\beta(H)^{-1}
   \operatorname{Tr}\bigl(e^{-\beta H_{0,n}}(D_A^\dagger D_A)^j\bigr).
 \end{aligned}
\]
The lower bound in \eqref{eq:bounded-partition-comparison} gives
\[
 \mathcal Z_\beta(H)^{-1}
 \le e^{\beta b}(2\sinh\beta)^{dn}.
\]
All modes other than $r$ cancel against the free partition function.
Since the one-mode occupation distribution is
$(1-q)q^m$, $q=e^{-2\beta}$, we obtain, for $j=1,2$,
\begin{align}
 \operatorname{Tr}\bigl(\rho_H(D_A^\dagger D_A)^j\bigr)
 &\le e^{\beta b}(1-q)\sum_{m\ge K}(m+1)^j q^m\nonumber\\
 &\le e^{\beta b}c_j(K+1)^j q^K.
 \label{eq:exact-tail-moments}
\end{align}
Here $(K+m+1)^j\le(K+1)^j(m+1)^j$, and summing the geometric
series and its first two derivatives gives $c_1,c_2$ in
\eqref{eq:direct-cutoff-constant}.

Let $\tau\le\mathfrak c\rho_H$ be a density operator.  Set
$J_{A,K}(s)=e^{isH}J_{A,K}e^{-isH}$ and
$d_A(s)=\widehat f(\theta_A)e^{i\theta_As}D_A$.
Equations \eqref{eq:exact-tail-evolution}--\eqref{eq:exact-tail-loss}
express the gain difference for this jump as
\[
 \int p_{\sigma_E}(s)\bigl(
 J_{A,K}(s)\tau d_A(s)^\dagger+
 d_A(s)\tau J_{A,K}(s)^\dagger+
 d_A(s)\tau d_A(s)^\dagger\bigr)\,ds.
\]
The loss difference is
$-|\widehat f(\theta_A)|^2(D_A^\dagger D_A\tau+\tau D_A^\dagger D_A)/2$.
We use $\|J_{A,K}(s)\|\le F\sqrt{K+1}$ and
$\int_{\mathbb R}p_{\sigma_E}(s)\,ds=1$.
For a bounded operator $J$ and an operator $D$ whose displayed moments
in $\tau$ are finite, Schatten Cauchy--Schwarz gives
\[
 \|J\tau D^\dagger\|_1\le\|J\|
       \operatorname{Tr}(\tau D^\dagger D)^{1/2},\qquad
 \|D^\dagger D\tau\|_1\le
       \operatorname{Tr}(\tau(D^\dagger D)^2)^{1/2}.
\]
The domination $\tau\le\mathfrak c\rho_H$ and
\eqref{eq:exact-tail-moments} imply
\[
 \operatorname{Tr}(\tau D_A^\dagger D_A)
 \le\mathfrak c e^{\beta b}c_1(K+1)q^K,
\]
and
\[
 \operatorname{Tr}(\tau(D_A^\dagger D_A)^2)
 \le\mathfrak c e^{\beta b}c_2(K+1)^2q^K.
\]
Since $q^K=e^{-2\beta K}$, the two mixed gain terms satisfy
\[
 \begin{aligned}
 \|J_{A,K}(s)\tau d_A(s)^\dagger\|_1
 &\le F\sqrt{K+1}\,v_f
    \bigl(\mathfrak c e^{\beta b}c_1(K+1)e^{-2\beta K}\bigr)^{1/2}\\
 &=Fv_f\sqrt{c_1}\,(K+1)
   \sqrt{\mathfrak c}\,e^{\beta b/2}e^{-\beta K},
 \end{aligned}
\]
and the adjoint mixed term obeys the same estimate.  The quadratic gain
term is bounded by
\[
 \|d_A(s)\tau d_A(s)^\dagger\|_1
 \le v_f^2\mathfrak c e^{\beta b}c_1(K+1)e^{-2\beta K}.
\]
For the anticommutator in the loss term,
\[
 \|D_A^\dagger D_A\tau\|_1
 \le\sqrt{\mathfrak c e^{\beta b}c_2}\,(K+1)e^{-\beta K},
\]
and the same bound holds for $\tau D_A^\dagger D_A$.
Summing these estimates over the $2dn$ jumps and using
$\sqrt{\mathfrak c}\le\mathfrak c$, $e^{\beta b/2}\le e^{\beta b}$,
and $e^{-2\beta K}\le e^{-\beta K}$ gives
\begin{align*}
 \|(\mathcal L_{\sigma_E,H}-\mathcal L_K)(\tau)\|_1
 &\le2dn(K+1)(2Fv_f\sqrt{c_1}+v_f^2\sqrt{c_2})
       \sqrt{\mathfrak c}\,e^{\beta b/2}e^{-\beta K}\\
 &\quad+2dn(K+1)v_f^2c_1\mathfrak c e^{\beta b}e^{-2\beta K}\\
 &\le C_{\beta,f,d}\,n\mathfrak c e^{\beta b}(K+1)e^{-\beta K}.
\end{align*}
The products here are trace class by \eqref{eq:exact-tail-moments}.
The bound $\tau\le\mathfrak c\rho_H$ also gives finite moments of
every power of $H$.  Schatten Cauchy--Schwarz then places $\tau$ in the
generator domain stated in Proposition~\ref{prop:sampler-generation}.

Set
\[
 \tau_s=e^{s\mathcal L_{\sigma_E,H}}\rho_{\rm ini}.
\]
Since the semigroup is positive and preserves $\rho_H$,
\[
 0\le\rho_{\rm ini}\le\mathfrak c\rho_H
 \quad\Longrightarrow\quad
 0\le\tau_s
 \le\mathfrak c e^{s\mathcal L_{\sigma_E,H}}\rho_H
 =\mathfrak c\rho_H.
\]
Thus the generator estimate above applies uniformly to every $\tau_s$.
The domain statement above and the boundedness of $\mathcal L_K$
justify Duhamel's formula in trace norm:
\[
 (e^{t\mathcal L_{\sigma_E,H}}-e^{t\mathcal L_K})(\rho_{\rm ini})
 =\int_0^t e^{(t-s)\mathcal L_K}
       (\mathcal L_{\sigma_E,H}-\mathcal L_K)(\tau_s)\,ds.
\]
Since $e^{(t-s)\mathcal L_K}$ is trace-norm contractive,
\[
 \begin{aligned}
 &\|(e^{t\mathcal L_{\sigma_E,H}}-e^{t\mathcal L_K})(\rho_{\rm ini})\|_1\\
 &\quad\le\int_0^t
 \|(\mathcal L_{\sigma_E,H}-\mathcal L_K)(\tau_s)\|_1\,ds\\
 &\quad\le t\,C_{\beta,f,d}
 n\mathfrak c e^{\beta b}(K+1)e^{-\beta K}.
 \end{aligned}
\]
This is \eqref{eq:direct-dynamical-cutoff}.
\end{proof}

\begin{theorem}
\label{thm:gibbs-preparation}
Under the input assumptions in Section~\ref{sec:algorithm-input}, let
$0<\eps<1$, $n,M\ge1$, and $\alpha_n^{\max}\le\operatorname{poly}(n)$.
Given
\[
 0<\gamma\le\min\{1,\operatorname{gap}(L_{\sigma_E,H_{n,M}})\},
\]
we can prepare a state within $\eps$ of $\sigma_\beta(H_{n,M})$ in trace
norm.  Writing $Q_{\rm sys}$ for the number of system-register qubits,
the construction uses
\begin{equation}
\label{eq:preparation-system-qubits}
 Q_{\rm sys}=\mathcal O\left(n\left[
 \log(2+nM)+\log\log(e^2+1/(\gamma\eps))\right]\right)
\end{equation}
system qubits and circuit depth
\[
 \widetilde{\mathcal O}\left(\gamma^{-1}
                   \operatorname{poly}(n,M,\log(2/\eps))\right).
\]
Simulation ancillas are given by \eqref{eq:circuit-ancillas}.
Taking
\begin{equation}
\label{eq:state-interaction-cutoff}
 M\asymp1+\left(\frac{\widehat\Xi_n}{\eps^2}\right)^{2d}
\end{equation}
with a sufficiently large constant gives preparation of
$\sigma_\beta(H_n)$, with the gap certificate supplied at this cutoff.
For $\alpha_n^{\max}\lesssim1$, it suffices to take
$M\asymp1+(n^{9/4}/\eps)^{4d}$.
\end{theorem}

\begin{proof}
We put $V=W_{n,M}$ and $b=\|V\|$, and use the certified upper bound
$\overline b$ in \eqref{eq:certified-input-bounds}.  By the triangle inequality and the global pair factorization,
\[
 \|V\|=\|W_{n,M}\|
 \le\sum_{i<j}|\alpha_{n,i,j}|\,\|\Pi_MW_{ij}\Pi_M\|
 =B_nv_M.
\]
The pair-matrix bound \eqref{eq:pair-matrix-bound} gives
\begin{equation}
\label{eq:algorithm-interaction-norm}
 b\le B_n v_M\le C_d B_n\sqrt{\lambda_M}
 \le\operatorname{poly}(n,M).
\end{equation}
Starting from $\rho_{\rm ini}=\kb{0}^{\otimes dn}$,
Lemma~\ref{lem:finite-rank-warm-start} gives the computable choice
\[
 \mathfrak c=e^{\beta(2\overline b+n\lambda_M)}(2\sinh\beta)^{-dn},
 \qquad \log\mathfrak c=\mathcal O_\beta(n\lambda_M+\overline b).
\]
We compute a rational upper bound for $\log\mathfrak c$; there is no
need to store the exponentially large number $\mathfrak c$.
By \eqref{eq:warm-start-mixing},
\[
 \|e^{t\mathcal L_{\sigma_E,H}}(\rho_{\rm ini})-\rho_H\|_1
 \le\mathfrak c e^{-\gamma t}.
\]
Thus the choice
\begin{equation}
\label{eq:algorithm-mixing-time}
 t=\gamma^{-1}\log(3\mathfrak c/\eps)
\end{equation}
gives
\[
 \mathfrak c e^{-\gamma t}=\frac\eps3.
\]
Hence the mixing error is at most $\eps/3$.

We next use Lemma~\ref{lem:direct-dynamical-cutoff}.
For $x\ge0$,
\[
 (x+1)e^{-\beta x/2}
 \le1+\sup_{y\ge0}ye^{-\beta y/2}
 =1+\frac2{e\beta}=:c_\beta.
\]
Therefore
\[
 (K+1)e^{-\beta K}
 \le c_\beta e^{-\beta K/2}.
\]
The bound in Lemma~\ref{lem:direct-dynamical-cutoff} is thus at most
\[
 c_\beta C_{\beta,f,d}nt\mathfrak c e^{\beta\overline b}
 e^{-\beta K/2}.
\]
Requiring this quantity to be at most $\eps/3$ gives the explicit choice
\begin{equation}
\label{eq:algorithm-occupation-cutoff}
 K=\max\left\{j_M+2,\left\lceil\frac2\beta
 \log\max\left\{1,
 \frac{3c_\beta C_{\beta,f,d}nt\mathfrak c e^{\beta\overline b}}{\eps}
 \right\}\right\rceil\right\}
\end{equation}
makes the dynamical truncation error at most $\eps/3$.
We round its numerical evaluation upwards.
Taking logarithms in \eqref{eq:algorithm-occupation-cutoff} gives
\[
 K\le C_{\beta,d,f}\left(
 j_M+1+\log n+\log t+\log\mathfrak c
 +\overline b+\log\frac1\eps\right).
\]
From \eqref{eq:algorithm-mixing-time},
\[
 \log t
 \le C+\log\frac1\gamma+
 \log\bigl(1+\log\mathfrak c+\log(1/\eps)\bigr).
\]
The warm-start estimate gives
\[
 \log\mathfrak c
 =\beta(2\overline b+n\lambda_M)-dn\log(2\sinh\beta)
 =\mathcal O_\beta(n\lambda_M+\overline b+n).
\]
Since $j_M\le C_d\lambda_M$, these bounds imply
\begin{equation}
\label{eq:algorithm-cutoff-scaling}
 K=\mathcal O_{\beta,d,f}\left(n\lambda_M+\overline b
                      +\log\frac{2+n}{\gamma\eps}\right).
\end{equation}
The evolution generated by $\mathcal L_K$ stays in the system
register.  Its circuit therefore uses $H_{\le K}$ and the truncated
ladders.

To simulate this finite-dimensional evolution, we apply
the Gaussian-quadrature approximation proposition and the corollary
``Circuit simulation of the quadrature generator'' in
\cite{BeckerRouzeSalzmannToAppearcmp} separately.  The jump block encodings have normalization $\sqrt{K+1}$, as checked
after \eqref{eq:circuit-mode-cutoff}.  The factor $\mathfrak c$ in
\eqref{eq:direct-dynamical-cutoff} affects $K$ only through
$\log\mathfrak c$.  The frequency cutoff in
\eqref{eq:algorithm-scalar-data} can be chosen with
$S=\mathcal O(1+dnK+dn+b)$.  The quadrature proposition gives an order
polynomial in $n,K,b,\log(2+t/\eps)$.  The circuit corollary and the
coherent scalar-access assumptions therefore simulate the finite evolution to error
$\eps/6$ at gate cost
\begin{equation}
\label{eq:finite-circuit-depth}
 \widetilde{\mathcal O}\left(t\,
    \operatorname{poly}(n,M,K,b,\log(2+t/\eps))\right).
\end{equation}

We describe the required circuits and check that they retain the
global interaction projection.  Each ladder on the system register
is a weighted shift with normalization $\sqrt{K+1}$; its matrix
entries can be precomputed and loaded with polynomial cost in $K$
and the requested precision.  The oscillator Hamiltonian is diagonal.
For a pair $i<j$, let $V_{M,ij}^{(2)}$ denote the matrix
$V_M^{(2)}$ acting on particles $i,j$, extended by zero outside their
pair cutoff and tensored with the identity on all other particles.  Then
\begin{equation}
\label{eq:global-pair-circuit}
 \Pi_M W_{ij}\Pi_M
 = V_{M,ij}^{(2)}\prod_{r\ne i,j}P_M^{(r)}.
\end{equation}
We test whether every particle outside the selected pair belongs to
$P_M$, apply the pair matrix only when all tests succeed, and reverse
the tests to reset the auxiliary registers.  A dense pair matrix has dimension
$M^2$, so a block encoding from the supplied matrix data has polynomial
cost in $M$ and precision.  Write the $M^2\times M^2$ pair matrix as
\[
 V_M^{(2)}=\sum_{a,b=1}^{M^2}v_{ab}|a\rangle\langle b|.
\]
Since $|v_{ab}|\le\|V_M^{(2)}\|=v_M$,
\[
 \sum_{a,b=1}^{M^2}|v_{ab}|
 \le M^4v_M.
\]
For the full interaction, linear-combination-of-unitaries normalization
therefore contributes at most
\[
 \sum_{i<j}|\alpha_{n,i,j}|M^4v_M=B_nM^4v_M.
\]
The truncated oscillator Hamiltonian contributes $\mathcal O(dnK+dn)$.
Combining the two parts gives a block encoding with a positive
normalization factor $\Lambda_{\rm enc}$ satisfying
\[
 \Lambda_{\rm enc}\le C_d(1+nK+B_nM^4v_M).
\]
Each matrix unit is a partial isometry with a constant-ancilla unitary
dilation; the controls on the other particles retain that property.
Hamiltonian simulation from this block encoding
\cite{LowChuang2019} has query cost
$\mathcal O(\Lambda_{\rm enc}|u|+\log(1/\eta))$ at time $u$ and accuracy
$\eta$.  The preparation and selection of the polynomially many
matrix-unit terms have polynomial overhead in $n,M,K$ and precision.
We finally control the error from rounding the Hamiltonian data.
On the finite occupation register, consider Hermitian matrices $T,T'$,
set $\delta_T=\|T-T'\|$, and use the same truncated bare jumps for both.
We write $a_A=\|A^{\le K}\|$, $m_j=m_j(f)$, and $J_T^A$ for the jump
obtained by filtering $A^{\le K}$ with the Hamiltonian $T$.
For every $u\in\mathbb R$, Duhamel's formula gives
\[
 e^{iuT}-e^{iuT'}
 =i\int_0^u e^{i(u-v)T}(T-T')e^{ivT'}\,dv,
\]
so
\[
 \|e^{iuT}-e^{iuT'}\|\le |u|\delta_T.
\]
Therefore
\[
 \begin{aligned}
 &\|e^{iuT}A^{\le K}e^{-iuT}
      -e^{iuT'}A^{\le K}e^{-iuT'}\|\\
 &\qquad\le
 \|(e^{iuT}-e^{iuT'})A^{\le K}e^{-iuT}\|
 +\|e^{iuT'}A^{\le K}(e^{-iuT}-e^{-iuT'})\|\\
 &\qquad\le2|u|a_A\delta_T.
 \end{aligned}
\]
Integrating against $|f(u)|$ yields
\[
 \|J_T^A-J_{T'}^A\|
 \le2m_1a_A\delta_T.
\]
Also
\[
 \|J_T^A\|,\ \|J_{T'}^A\|
 \le\int|f(u)|a_A\,du=m_0a_A.
\]
Thus
\[
 \begin{aligned}
 &(J_T^A)^\dagger J_T^A-(J_{T'}^A)^\dagger J_{T'}^A\\
 &\quad=(J_T^A-J_{T'}^A)^\dagger J_T^A
 +(J_{T'}^A)^\dagger(J_T^A-J_{T'}^A),
 \end{aligned}
\]
and hence
\[
 \|(J_T^A)^\dagger J_T^A-(J_{T'}^A)^\dagger J_{T'}^A\|
 \le4m_0m_1a_A^2\delta_T.
\]
The kernels in \eqref{eq:loss-integral} and
\eqref{eq:smooth-gaussian-average} have finite moments.  We define
\[
 \zeta_j=\int_{\mathbb R}|s|^j|g_{\sigma_E}(s)|\,ds,
 \qquad j=0,1,
\]
and
\[
 \pi_1=\int_{\mathbb R}|s|p_{\sigma_E}(s)\,ds.
\]
These quantities depend only on the fixed filter parameters.  Set
\[
 C_{\rm data}
 =8m_0m_1\zeta_0+4m_0^2\zeta_1+4m_0m_1+4m_0^2\pi_1.
\]
Let $G_{T,K}$ and $\Phi_{T,K}$ denote the loss operator and gain map
in \eqref{eq:loss-integral}--\eqref{eq:gain-integral}, with $H=T$ and
each jump replaced by $A^{\le K}$.  We write $\mathcal L_{T,K}$ for
the resulting finite sampler.  For a linear map $\mathcal T$ on
trace-class operators, its induced trace norm is
$\|\mathcal T\|_{1\to1}=\sup_{\|X\|_1=1}\|\mathcal T(X)\|_1$.
Applying the same conjugation estimate to the outer time evolutions
in the gain and loss terms gives
\[
 \begin{aligned}
 \|G_{T,K}-G_{T',K}\|
 &\le\delta_T\sum_{A\in\mathcal A}a_A^2(4m_0m_1\zeta_0+2m_0^2\zeta_1),\\
 \|\Phi_{T,K}-\Phi_{T',K}\|_{1\to1}
 &\le4\delta_T\sum_{A\in\mathcal A}a_A^2(m_0m_1+m_0^2\pi_1),\\
 \|\mathcal L_{T,K}-\mathcal L_{T',K}\|_{1\to1}
 &\le C_{\rm data}\delta_T\sum_{A\in\mathcal A}a_A^2.
 \end{aligned}
\]
For bounded generators,
\[
 e^{t\mathcal L_{T,K}}-e^{t\mathcal L_{T',K}}
 =\int_0^t
 e^{(t-s)\mathcal L_{T,K}}
 (\mathcal L_{T,K}-\mathcal L_{T',K})
 e^{s\mathcal L_{T',K}}\,ds.
\]
Both semigroups are trace-norm contractive, so
\[
 \|e^{t\mathcal L_{T,K}}-e^{t\mathcal L_{T',K}}\|_{1\to1}
 \le t\|\mathcal L_{T,K}-\mathcal L_{T',K}\|_{1\to1}.
\]
Every truncated ladder has norm at most $\sqrt{K+1}$, and there are
$2dn$ of them.  Thus
\[
 \sum_{A\in\mathcal A}a_A^2\le2dn(K+1).
\]
Combining these estimates gives
\begin{equation}
\label{eq:input-channel-stability}
 \|e^{t\mathcal L_{T,K}}-e^{t\mathcal L_{T',K}}\|_{1\to1}
 \le C_{\rm data}t\delta_T\sum_{A\in\mathcal A}a_A^2
 \le2C_{\rm data}tdn(K+1)\delta_T.
\end{equation}
We therefore round $T=H_{\le K}$ to a Hermitian matrix $T'$ with
\[
 \|T-T'\|
 \le \min\left\{1,
 \frac{\eps}{12C_{\rm data}tdn(K+1)}\right\}.
\]
Equation \eqref{eq:input-channel-stability} makes this error at most
$\eps/6$.  The input model supplies this accuracy with logarithmic
dependence on $t$ and $\eps^{-1}$, in addition to the polynomial matrix
and coupling bounds.  For the rounded matrix, we may use
$\overline b+1$ in the circuit norm bounds.  Simulating its evolution
to error $\eps/6$ gives a total finite-simulation error at most $\eps/3$.
The gap certificate is used only for the ideal sampler; no gap estimate
for the rounded Hamiltonian is needed.

The mixing, occupation-cutoff, and finite-simulation errors sum to
$\eps$.  Equations
\eqref{eq:algorithm-interaction-norm},
\eqref{eq:algorithm-mixing-time},
\eqref{eq:algorithm-cutoff-scaling}, and
\eqref{eq:finite-circuit-depth} give the depth bound, with powers of
$\log(2+1/\gamma)$ absorbed in $\widetilde{\mathcal O}$.
Encoding the $dn$ modes uses $dn\lceil\log_2(K+1)\rceil$ qubits,
which proves \eqref{eq:preparation-system-qubits}.
The additional simulation workspace is
\eqref{eq:circuit-ancillas}, with $\eta$ a fixed fraction of $\eps$.
To prepare the Gibbs state of the full Coulomb Hamiltonian, we refine
the coupling enclosure to
$0\le\overline\Xi_n-\widehat\Xi_n\le\eps^2$ and choose
\[
 M=\max\left\{1,\left\lceil
 \left(\frac{4C_\rho^2\overline\Xi_n}{\eps^2}\right)^{2d}
 \right\rceil\right\}.
\]
By construction,
\[
 M^{1/(2d)}
 \ge\frac{4C_\rho^2\overline\Xi_n}{\eps^2}.
\]
Hence
\[
 C_\rho \overline\Xi_n^{1/2}M^{-1/(4d)}
 \le\frac\eps2.
\]
The inequality
\[
 \overline\Xi_n\le\widehat\Xi_n+\eps^2
\]
gives the order stated in \eqref{eq:state-interaction-cutoff}.
Theorem~\ref{thm:coulomb-approximation} therefore makes the interaction
truncation error at most $\eps/2$ for this choice.  Applying finite-rank preparation
at accuracy $\eps/2$ proves the assertion for $\sigma_\beta(H_n)$.
When $\alpha_n^{\max}\lesssim1$, $\widehat\Xi_n\lesssim n^{9/2}$,
which gives the last cutoff choice.  For bounded-degree interactions
with bounded couplings, the improved coefficient instead gives
$M=\mathcal O(1+n^{3d}\eps^{-4d})$; for free energy it gives
$M=\mathcal O(1+n^{3d}\eps^{-2d})$.  These rank improvements still
require the gap certificate at the selected rank.  Substitution into
\eqref{eq:preparation-system-qubits} proves the Gibbs-state part of
Theorem~\ref{thm:algorithms-intro}.
\end{proof}

\subsection{Estimating free energies}
\smallsection{Randomized thermodynamic integration}
\label{sec:free-energy-implementation}
For $s\in[0,1]$, we write $V=W_{n,M}$, $H(s)=H_{0,n}+sV$,
$\rho_s=\sigma_\beta(H(s))$, and $e(s)=\operatorname{Tr}(\rho_sV)$.
A prime on $e$ denotes differentiation in $s$.
The reference free energy and thermodynamic integration identity are
\begin{equation}
\label{eq:thermodynamic-integration}
 F_\beta(H_{0,n})=\frac{dn}{\beta}\log(2\sinh\beta),\qquad
 F_\beta(H_{n,M})-F_\beta(H_{0,n})=\int_0^1e(s)\,ds.
\end{equation}
Boundedness of $V$ and the heat-trace bounds justify Duhamel
differentiation in trace norm.  First,
\[
 \partial_s e^{-\beta H(s)}
 =-\beta\int_0^1
 e^{-(1-u)\beta H(s)}V e^{-u\beta H(s)}\,du.
\]
Taking the trace gives
\[
 \partial_s\mathcal Z_\beta(H(s))
 =-\beta\mathcal Z_\beta(H(s))e(s).
\]
Differentiating
$e(s)=\mathcal Z_\beta(H(s))^{-1}
\operatorname{Tr}(Ve^{-\beta H(s)})$ therefore yields
\[
 \begin{aligned}
 e'(s)
 &=-\beta\int_0^1
 \operatorname{Tr}\bigl(\rho_s^{1-u}V\rho_s^uV\bigr)\,du
 +\beta e(s)^2\\
 &=-\beta\int_0^1
 \operatorname{Tr}\bigl(
 \rho_s^{1-u}(V-e(s))\rho_s^u(V-e(s))\bigr)\,du\\
 &=-\beta\int_0^1
 \bigl\|\rho_s^{(1-u)/2}(V-e(s))\rho_s^{u/2}\bigr\|_2^2\,du
 \le0.
 \end{aligned}
\]
Thus, for every integer $L\ge1$, the left Riemann sum satisfies
\begin{equation}
\label{eq:RiemannConclusion}
 0\le\frac1L\sum_{k=0}^{L-1}e(k/L)-\int_0^1e(s)\,ds
 \le\frac{e(0)-e(1)}L\le\frac{2\|V\|}L.
\end{equation}
For $s\in[k/L,(k+1)/L]$, monotonicity gives
\[
 0\le e(k/L)-e(s)
 \le e(k/L)-e((k+1)/L).
\]
Integrating over this interval,
\[
 0\le
 \frac1Le(k/L)-\int_{k/L}^{(k+1)/L}e(s)\,ds
 \le\frac1L\bigl(e(k/L)-e((k+1)/L)\bigr).
\]
Summing over $k=0,\ldots,L-1$ telescopes:
\[
 0\le\frac1L\sum_{k=0}^{L-1}e(k/L)-\int_0^1e(s)\,ds
 \le\frac{e(0)-e(1)}L.
\]
Finally $|e(s)|\le\|V\|$, so
\[
 e(0)-e(1)\le2\|V\|.
\]

\begin{theorem}
\label{thm:free-energy-algorithm}
Under the input assumptions in Section~\ref{sec:algorithm-input}, let
$0<\eps<1$, $0<\delta\le1/2$, and
$\alpha_n^{\max}\le\operatorname{poly}(n)$.  Choose
\begin{equation}
\label{eq:algorithm-interaction-cutoff}
 M\asymp1+\left(\frac{\widehat\Xi_n}{\eps}\right)^{2d}
\end{equation}
with a sufficiently large constant depending on $\beta,d$.
Given
\[
 0<\gamma_*\le\min\left\{1,
 \min_{0\le s\le1}\operatorname{gap}(L_{\infty,H(s)})\right\},
\]
we can estimate $F_\beta(H_n)$ to error $\eps$ with failure probability
at most $\delta$.  The algorithm uses
\[
 \mathcal O\left(n\left[\log(2+n/\eps)
                   +\log\log(e^2+1/\gamma_*)\right]\right)
\]
system qubits and runtime
\[
 \widetilde{\mathcal O}\left(\gamma_*^{-1}\log(2/\delta)
                               \operatorname{poly}(n,1/\eps)\right).
\]
For $v_M$ in \eqref{eq:uniform-coulomb-threshold} and $\Gamma=B_nv_M$,
we denote the total number of independent prepared states and selected
pair measurements by $N_{\rm samp}$.  It satisfies
\begin{equation}
\label{eq:total-free-energy-samples}
 N_{\rm samp}=\mathcal O\left(\left[1+\frac{\Gamma^2}{\eps^2}\right]
                              \log\frac2\delta\right).
\end{equation}
This counts all prepared states, not the states used for each pair or
each grid point.
The same quantum registers are reused; the simulation ancillas satisfy
\eqref{eq:circuit-ancillas}.
\end{theorem}

\begin{proof}
For $n=1$, the reference formula in \eqref{eq:thermodynamic-integration}
suffices.  For $n\ge2$, we allocate $\eps/6$ to each of interaction
truncation, quadrature, state preparation, sampling, numerical data,
and evaluation of the reference free energy.

Using the certified coupling bound from
Section~\ref{sec:algorithm-input}, we choose
\[
 M=\max\left\{1,\left\lceil
       \left(\frac{6C_F\overline\Xi_n}{\eps}\right)^{2d}
                         \right\rceil\right\}.
\]
If $\overline\Xi_n=0$, the interaction vanishes and the reference free
energy is exact; otherwise, by the definition of $M$,
\[
 M^{-1/(2d)}
 \le\frac\eps{6C_F\overline\Xi_n}.
\]
Hence Theorem~\ref{thm:coulomb-approximation} gives
\[
 |F_\beta(H_n)-F_\beta(H_{n,M})|
 \le C_F\overline\Xi_nM^{-1/(2d)}
 \le\frac\eps6.
\]
The bound
\[
 \overline\Xi_n\le\widehat\Xi_n+\eps
\]
gives the order in \eqref{eq:algorithm-interaction-cutoff}.
For $\overline\Gamma=\overline B_n\overline v_M\ge\|V\|$, we take
\[
 L=2^{\left\lceil\log_2\max\{1,12\overline\Gamma/\eps\}\right\rceil},
 \qquad s_k=k/L\quad(0\le k<L).
\]
By this dyadic definition of $L$,
\[
 L\ge\frac{12\overline\Gamma}{\eps}.
\]
Since $\|V\|\le\overline\Gamma$, \eqref{eq:RiemannConclusion} gives
\[
 0\le\frac1L\sum_{k=0}^{L-1}e(s_k)-\int_0^1e(s)\,ds
 \le\frac{2\|V\|}{L}
 \le\frac\eps6.
\]
We sample $k$ uniformly; we do not visit every grid point.

We use $\mathbb E$ for expectation and $\mathbb P$ for probability over
the classical sampling and measurement outcomes.  The function
$\operatorname{sgn}$ is the sign, with $\operatorname{sgn}(0)=0$.
For exact data, define $O_{ij}=\Pi_MW_{ij}\Pi_M$.  When $B_n>0$, choose a pair with probability
\[
 p_{ij}=\frac{|\alpha_{n,i,j}|}{B_n}.
\]
At grid point $s_k$, measure $O_{ij}$ in the state $\rho_{s_k}$ and
denote the outcome by $Y$.  Conditional on $k,i,j$,
\[
 \mathbb E[Y\mid k,i,j]
 =\operatorname{Tr}(\rho_{s_k}O_{ij}).
\]
Return
\[
 Z=B_n\operatorname{sgn}(\alpha_{n,i,j})Y.
\]
Averaging first over the pair gives
\[
 \begin{aligned}
 \mathbb E[Z\mid k]
 &=\sum_{i<j}\frac{|\alpha_{n,i,j}|}{B_n}
 B_n\operatorname{sgn}(\alpha_{n,i,j})
 \operatorname{Tr}(\rho_{s_k}O_{ij})\\
 &=\operatorname{Tr}\left(
 \rho_{s_k}\sum_{i<j}\alpha_{n,i,j}O_{ij}\right)
 =\operatorname{Tr}(\rho_{s_k}V).
 \end{aligned}
\]
Choosing $k$ uniformly and independently then gives
\begin{equation}
\label{eq:joint-estimator}
 \mathbb EZ=\frac1L\sum_{k=0}^{L-1}\operatorname{Tr}(\rho_{s_k}V),
 \qquad |Z|\le B_nv_M=\Gamma.
\end{equation}
The case $B_n=0$ requires no sampling.

We next account for finite-precision data in \eqref{eq:joint-estimator}.
We obtain real dyadic couplings $\widetilde\alpha_{ij}$ (rationals
with a power-of-two denominator) and Hermitian observables
$\widetilde{\mathcal O}_{ij}$ with the same global support.  Define the observable
and coupling tolerances
\begin{equation}
\label{eq:randomized-data-precision}
 \delta_O:=\frac{\eps}{24(1+\overline B_n)},\qquad
 \delta_\alpha:=\frac{\eps}{24(1+\overline v_M)}.
\end{equation}
We choose the data so that
$\|\widetilde{\mathcal O}_{ij}-O_{ij}\|\le\delta_O$ for every pair and
$\sum_{i<j}|\widetilde\alpha_{ij}-\alpha_{n,i,j}|\le\delta_\alpha$.
We construct all these observables from one approximating pair matrix
and the projections on the other particles.  To specify the observable actually measured, we compute an approximate
pair-matrix diagonalization with a bound on its residual, compile the
basis change into a unitary $\widetilde U$ on the encoded pair space,
and choose dyadic output values $y_j$, $1\le j\le M^2$.  We extend the compiled unitary
by the identity outside the encoded pair space.  We write $\operatorname{diag}(y_1,\ldots,y_{M^2})$ for the diagonal
matrix with diagonal entries $y_1,\ldots,y_{M^2}$.  We define the
output matrix and the measured pair observable by
\[
 D_{\rm out}=\operatorname{diag}(y_1,\ldots,y_{M^2}),\qquad
 \widetilde V_M^{(2)}=\widetilde U^\dagger D_{\rm out}\widetilde U.
\]
Applying $\widetilde U$, measuring in the computational basis, and
returning $y_j$ measures this operator.  For a target basis change $U$, insert and subtract
$U^\dagger D_{\rm out}\widetilde U$ to obtain
\[
 \begin{aligned}
 \|\widetilde U^\dagger D_{\rm out}\widetilde U
       -U^\dagger D_{\rm out}U\|&\le
 \|(\widetilde U^\dagger-U^\dagger)D_{\rm out}\widetilde U\|
 +\|U^\dagger D_{\rm out}(\widetilde U-U)\|\\
 &\quad\le2\|D_{\rm out}\|\,\|\widetilde U-U\|.
 \end{aligned}
\]
We choose the diagonalization residual, output rounding, and unitary
synthesis error so that
$\|\widetilde V_M^{(2)}-V_M^{(2)}\|\le \delta_O$.
Extending by zero outside the pair cutoff and retaining the projections
on all other particles gives $\widetilde{\mathcal O}_{ij}$ in
\eqref{eq:randomized-data-precision}.
This construction controls both the measurement mean and its output
range.  It does not require a lower bound on the spacing between pair
matrix eigenvalues.
We set
\[
 \widetilde B=\sum_{i<j}|\widetilde\alpha_{ij}|,\qquad
 \widetilde V=\sum_{i<j}\widetilde\alpha_{ij}\widetilde{\mathcal O}_{ij},\qquad
 \widetilde v=\max_{i<j}\|\widetilde{\mathcal O}_{ij}\|.
\]
The coupling approximation gives
\[
 \widetilde B
 \le\sum_{i<j}|\alpha_{n,i,j}|+
     \sum_{i<j}|\widetilde\alpha_{ij}-\alpha_{n,i,j}|
 \le B_n+\delta_\alpha,
\]
and the observable approximation gives
\[
 \widetilde v\le v_M+\delta_O\le\overline v_M+1.
\]
Moreover,
\[
 \begin{aligned}
 \widetilde V-V
 &=\sum_{i<j}\alpha_{n,i,j}(\widetilde{\mathcal O}_{ij}-O_{ij})
   +\sum_{i<j}(\widetilde\alpha_{ij}-\alpha_{n,i,j})\widetilde{\mathcal O}_{ij},
 \end{aligned}
\]
so
\[
 \begin{aligned}
 \|\widetilde V-V\|
 &\le B_n\delta_O+\delta_\alpha\widetilde v\le B_n\delta_O+\delta_\alpha(v_M+\delta_O).
 \end{aligned}
\]
Similarly,
\[
 \begin{aligned}
 \widetilde B\widetilde v
 &\le(B_n+\delta_\alpha)(v_M+\delta_O)=\Gamma+B_n\delta_O+\delta_\alpha v_M+\delta_\alpha\delta_O.
 \end{aligned}
\]
With the choices in \eqref{eq:randomized-data-precision}, these estimates
become
\begin{equation}
\label{eq:randomized-observable-error}
 \begin{aligned}
 \|\widetilde V-V\|&\le B_n\delta_O+\delta_\alpha(v_M+\delta_O)\le\eps/12,\\
 \widetilde B\widetilde v&\le\Gamma+B_n\delta_O+\delta_\alpha(v_M+\delta_O)
                         \le\Gamma+\eps/12.
 \end{aligned}
\end{equation}
We compute a rational outcome bound $G$ satisfying
\[
 \widetilde B\widetilde v\le G
 \le\widetilde B\widetilde v+\frac\eps{12}
 \le\Gamma+\frac\eps6.
\]
If $\widetilde B=0$, then every $\widetilde\alpha_{ij}=0$, so
\[
 B_n
 \le\sum_{i<j}|\alpha_{n,i,j}-\widetilde\alpha_{ij}|
 \le\delta_\alpha.
\]
Therefore
\[
 \|V\|\le B_nv_M
 \le\delta_\alpha v_M
 \le\frac\eps{24}<\frac\eps{12}.
\]
The thermodynamic-integration identity then gives
\[
 |F_\beta(H_{n,M})-F_\beta(H_{0,n})|
 \le\int_0^1|e(s)|\,ds
 \le\|V\|<\frac\eps{12}.
\]
Thus the reference free energy alone has sufficient accuracy for the
interaction-truncated problem, and we stop without preparing any states.

Otherwise, the ideal pair probabilities are
$p_{ij}=|\widetilde\alpha_{ij}|/\widetilde B$.  We round them to dyadic
probabilities $p'_{ij}$ summing to one and satisfying
\[
 \sum_{i<j}|p'_{ij}-p_{ij}|\le\frac{\eps}{12(1+G)}.
\]
Rounding a probability vector of at most $n^2$ entries requires only
$\mathcal O(\log(2+n)+\log(2+(1+G)/\eps))$ bits.  It can be done by rounding
cumulative probabilities, and we can sample the resulting distribution with a fixed finite number
of unbiased random bits.  The dyadic choice of $L$ makes grid
sampling exact as well.

At the selected grid point, define
\[
 \eta=\frac\eps{6(1+G)}.
\]
Theorem~\ref{thm:gibbs-preparation} prepares an independent state
$\widetilde\rho_k$ satisfying
\[
 \|\widetilde\rho_k-\rho_{s_k}\|_1\le\eta.
\]
The same lower bound $\gamma_*$ applies to each finite-$\sigma_E$ sampler
by \eqref{eq:smooth-gaussian-average}.  Positivity of the path minimum
is Lemma~\ref{lem:path-gap-continuity}; in the regime of
Theorem~\ref{thm:uniform-gap}, we can use
$\gamma_*=\min\{1,\lambda_{\rm OU}/4\}$.

We measure $\widetilde{\mathcal O}_{ij}$, call its outcome $Y$, and return
$\widetilde Z=\widetilde B\operatorname{sgn}(\widetilde\alpha_{ij})Y$.
All outcomes satisfy
\[
 |\widetilde Z|
 \le\widetilde B\widetilde v\le G.
\]
With the ideal probabilities $p_{ij}$, the same calculation as above
gives
\[
 \mathbb E\widetilde Z
 =\frac1L\sum_{k=0}^{L-1}
 \operatorname{Tr}(\widetilde\rho_k\widetilde V).
\]
The state-preparation error satisfies
\[
 \left|\operatorname{Tr}((\widetilde\rho_k-\rho_{s_k})\widetilde V)\right|
 \le\eta\|\widetilde V\|.
\]
Changing the pair distribution from $p_{ij}$ to $p'_{ij}$ changes the
mean by at most
\[
 G\sum_{i<j}|p'_{ij}-p_{ij}|.
\]
Together with the observable error from
\eqref{eq:randomized-observable-error}, these bounds imply
\[
 \begin{aligned}
 \left|\mathbb E\widetilde Z
       -\frac1L\sum_k\operatorname{Tr}(\rho_{s_k}V)\right|
 &\le\eta\|\widetilde V\|+\|\widetilde V-V\|
                 +G\sum_{i<j}|p'_{ij}-p_{ij}|\\
 &\le\frac\eps6+\frac\eps{12}+\frac\eps{12}.
 \end{aligned}
\]
The first term is the state-preparation budget; the last two terms are
the numerical-data budget.

For an integer $N\ge1$, let $\widetilde Z_1,\ldots,\widetilde Z_N$
be independent repetitions of $\widetilde Z$.  Hoeffding's inequality
\cite{Hoeffding1963} gives
\[
 \mathbb P\left(\left|\frac1N\sum_{m=1}^N\widetilde Z_m
                           -\mathbb E\widetilde Z\right|>\frac\eps6\right)
 \le2\exp\left(-\frac{N\eps^2}{72G^2}\right)
 \quad(G>0).
\]
Thus the choice
\begin{equation}
\label{eq:randomized-sample-count}
 N=\max\left\{1,\left\lceil
       \frac{72G^2}{\eps^2}\log\frac2\delta\right\rceil\right\}
\end{equation}
proves \eqref{eq:total-free-energy-samples}.  For $G=0$ there is no
sampling error.  No union bound over pairs or grid points is needed.

The measurements retain the global factorization
\eqref{eq:global-pair-circuit}.  We test membership in $P_M$ on every
particle register.  When all tests succeed, we measure the approximating
pair matrix; otherwise we return zero.  This realizes a spectral
measurement of $\widetilde{\mathcal O}_{ij}$, including its zero eigenspace outside
$\Pi_M$.  Computing the approximate diagonalization and compiling its basis
change have polynomial classical and gate cost in $M$ and the requested
precision \cite[Section~4.5.1]{Nielsen_Chuang_2010}.  The membership
controls have polynomial overhead in $n,M,K$.  The entrywise-to-operator
bound after \eqref{eq:algorithm-pair-data} makes the precision in
\eqref{eq:randomized-data-precision} attainable with polynomially many
bits.  Hamiltonian input errors are part of the state-preparation budget and
are controlled by \eqref{eq:input-channel-stability}.

Finally, we evaluate the elementary reference free energy in
\eqref{eq:thermodynamic-integration} to error $\eps/6$.  On the
sampling-success event, the total error is bounded by
\[
 \begin{aligned}
 &\underbrace{\eps/6}_{\text{interaction cutoff}}
 +\underbrace{\eps/6}_{\text{quadrature}}
 +\underbrace{\eps/6}_{\text{state preparation}}
 +\underbrace{\eps/6}_{\text{sampling}}\\
 &\qquad
 +\underbrace{\eps/6}_{\text{numerical data}}
 +\underbrace{\eps/6}_{\text{reference free energy}}
 =\eps.
 \end{aligned}
\]
The quantities $M,\overline\Gamma,G,\eta^{-1}$ are bounded by
polynomials in
$n,1/\eps$.  Theorem~\ref{thm:gibbs-preparation} gives cost
$\widetilde{\mathcal O}(\gamma_*^{-1}\operatorname{poly}(n,1/\eps))$ per sample.
Multiplying this cost by \eqref{eq:randomized-sample-count} proves the
runtime bound.  There is no additional factor for the number of pairs
or grid points.
The system-qubit bound follows from
\eqref{eq:preparation-system-qubits}; \eqref{eq:circuit-ancillas} applies
with the per-sample time and accuracy.  This also proves the free-energy
part of Theorem~\ref{thm:algorithms-intro}.
\end{proof}

\smallsection{Acknowledgement}
SB acknowledges support from SNF Grant PZ00P2\_216019.
CR is supported by France 2030 under the
French National Research Agency award number ''ANR-22-EXES-0013''. RS acknowledges support by the European Research Council (ERC Grant Agreement No.~948139) and from the Excellence Cluster Matter and Light for Quantum Computing (ML4Q-2).

\smallsection{AI statement}
AI-assisted tools were used only for light language editing in the first version of this article posted on arXiv. In preparing the revised version, we used these tools for a detailed critical review of the original manuscript, which helped us strengthen several statements and streamline the presentation. The authors are responsible for the mathematical content and the final text.

\providecommand{\bysame}{\leavevmode\hbox to3em{\hrulefill}\thinspace}
\providecommand{\MR}{\relax\ifhmode\unskip\space\fi MR }
\providecommand{\MRhref}[2]{%
  \href{http://www.ams.org/mathscinet-getitem?mr=#1}{#2}
}
\providecommand{\href}[2]{#2}

\end{document}